\documentclass[a4paper,fleqn]{cas-dc}

\usepackage{amsmath}
\usepackage{amsfonts}
\usepackage{amsthm}
\usepackage[caption=false]{subfig} 
\usepackage{makecell}
\newtheorem{theorem}{Theorem}

\usepackage[numbers,sort]{natbib}
\usepackage{algorithm}
\usepackage{algpseudocode}

\def\tsc#1{\csdef{#1}{\textsc{\lowercase{#1}}\xspace}}
\tsc{WGM}
\tsc{QE}
\tsc{EP}
\tsc{PMS}
\tsc{BEC}
\tsc{DE}

\begin{document}
\let\WriteBookmarks\relax
\def\floatpagepagefraction{1}
\def\textpagefraction{.001}
\shortauthors{Gao et~al.}

\title [mode = title]{Optimal Short Video Ordering and Transmission Scheduling for Reducing Video Delivery Cost in Peer-to-Peer CDNs}                      

%

\author[1]{Zhipeng Gao}
\ead{zhpgao@bjtu.edu.cn}
\credit{Investigation, Methodology, Data analysis, Simulation, Writing – original draft}

\author[1]{Chunxi Li}

\ead{chxli1@bjtu.edu.cn}
\credit{Investigation, Methodology, Data analysis, Writing – original draft}

\author[1]{Yongxiang Zhao}
\cormark[1]

\ead{yxzhao@bjtu.edu.cn}
\credit{Investigation, Methodology, Data analysis, Writing – review \& editing}


\affiliation[1]{
organization={School of Electronic and Information Engineering, Beijing Jiaotong University},
city={Beijing},
country={China}}



\cortext[cor1]{Corresponding author}
\begin{abstract}
The explosive growth of short video platforms has generated a massive surge in global traffic, imposing heavy financial burdens on content providers. While Peer-to-Peer Content Delivery Networks (PCDNs) offer a cost-effective alternative by leveraging resource-constrained edge nodes, the limited storage and concurrent service capacities of these peers struggle to absorb the intense temporal demand spikes characteristic of short video consumption. In this paper, we propose to minimize transmission costs by exploiting a novel degree of freedom, the inherent flexibility of server-driven playback sequences. We formulate the Optimal Video Ordering and Transmission Scheduling (OVOTS) problem as an Integer Linear Program to jointly optimize personalized video ordering and transmission scheduling. By strategically permuting playlists, our approach proactively smooths temporal traffic peaks, maximizing the offloading of requests to low-cost peer nodes. To solve the OVOTS problem, we provide a rigorous theoretical reduction of the OVOTS problem to an auxiliary Minimum Cost Maximum Flow (MCMF) formulation. Leveraging König's Edge Coloring Theorem, we prove the strict equivalence of these formulations and develop the Minimum-cost Maximum-flow with Edge Coloring (MMEC) algorithm, a globally optimal, polynomial-time solution. Extensive simulations demonstrate that MMEC significantly outperforms baseline strategies, achieving cost reductions of up to 67\% compared to random scheduling and 36\% compared to a simulated annealing approach. Our results establish playback sequence flexibility as a robust and highly effective paradigm for cost optimization in PCDN architectures.

\end{abstract}




\maketitle

\section{Introduction}
The explosive growth of short video platforms, such as TikTok and Instagram Reels, has revolutionized digital content consumption, leading to an unprecedented surge in global video traffic~\cite{huang2023digital,shang2025large}. To manage this surge, content providers have traditionally distributed video data through Content Delivery Networks (CDNs)~\cite{ghabashneh2020exploring,zhang2014unreeling}, incurring prohibitive costs to satisfy user demand. For instance, Douyin serves over 1 billion daily active users, with annual traffic costs exceeding 6 billion RMB~\cite{wei2023multipath}. To mitigate these substantial expenses, providers are increasingly adopting cost-effective Peer-to-peer Content Delivery Network (PCDN) architectures~\cite{wei2025qoe}. By integrating low-cost, resource-constrained peer nodes (e.g., home routers and set-top boxes) with robust but expensive traditional CDN nodes, providers can strategically offload a portion of video traffic to the peer network, thereby reducing overall transmission costs~\cite{wei2024pscheduler,farahani2023alive}.

However, reducing short video transmission costs within PCDN architectures remains challenging due to a fundamental mismatch between the bandwidth and storage limitations of low-cost peer nodes and the flash-crowd access pattern~\cite{jung2002flash} of short videos. On one hand, peer nodes are typically consumer-grade devices inherently restricted by limited upload bandwidth and storage capacity~\cite{zhang2024enhancing,wang2024twist}. For instance, over 80\% of these nodes provide an average serving speed of less than 5 Mbps, and fewer than 1\% exceed 8 Mbps~\cite{zhang2024enhancing}. Moreover, they usually rely on the residual storage of host devices, constraining their cache capacity to the order of gigabytes~\cite{liu2011novasky}. On the other hand, a short video often experiences massive concurrent requests within a brief timeframe, followed by a rapid decline in popularity~\cite{huang2023digital,huang2024digital,gao2022short}. According to our measurements in Section~\ref{background}, a short video typically loses its popularity within 30 minutes. Consequently, the concurrent service capabilities of these low-cost nodes are entirely insufficient to absorb such temporally concentrated request spikes. The excess request is inevitably forced back to costly CDN nodes, severely undermining the economic advantages of the PCDN architecture.

To reduce overall video transmission costs, current solutions primarily focus on proactive caching and optimizing transmission scheduling to maximize traffic offloading. First, proactive caching strategies~\cite{viola2020predictive,peroni2025end,tan2019online} heavily rely on predicting video popularity and then proactive cache updating. However, given the rapid shifts in short video popularity, this approach can lead to frequent cache replacements, quickly exhausting the limited I/O and bandwidth resources of peer nodes. Second, traditional transmission scheduling~\cite{zhang2024enhancing,wang2024twist,wei2025qoe} load-balances video requests across available peer nodes, while it is typically reactive to the user's fixed video viewing demand. Consequently, when confronting the massive, system-wide flash-crowd access pattern of short videos, both strategies yield limited effectiveness because they cannot alter the underlying temporal concentration of user demand.

In this paper, we approach the reduction of short video transmission costs from a novel perspective, exploiting the inherent flexibility of the short video playback sequence and transforming it into a new degree of freedom for optimization. This insight leverages the server-driven nature of short video platforms~\cite{gao2025startup,quan2023alleviating}, where a central server periodically pushes a predefined playlist to the user, determining the specific recommended videos, their playback sequence, and the designated delivery node for each video\footnote{While short video systems typically decouple recommendation engines from transmission schedulers, our framework tightly couples node designation with playlist generation to facilitate a mathematically tractable joint optimization.}~\cite{gao2022short,simonovski2025swiping,gong2022real}. Crucially, since this playback sequence is generated on the server side, the server can strategically permute the sequence of videos within a personalized playlist without altering
the actual set of recommended videos, which ensures that
the optimization remains transparent to the user. To ensure this optimization remains transparent to the user, the permutation is strictly confined to a set of equally weighted recommended content, thereby preserving the baseline Quality of Experience (QoE). Leveraging this hidden flexibility, we propose a joint optimization framework that simultaneously: 1) staggers the playback times of identical videos across different users to proactively smooth temporal demand spikes, and 2) intelligently assigns an optimal peer node to serve each video to minimize the overall transmission cost. By flattening the highly concentrated traffic peaks, our framework ensures that requests which would otherwise overflow to expensive CDN nodes can be effectively absorbed by low-cost peer nodes with limited concurrent service capacities.

To realize this approach, we abstract the core components of the short video delivery architecture, modeling it as a system comprising a centralized server acting as the controller, multiple short video users, a set of low-cost, capacity-limited peer nodes, and a robust but expensive CDN node to serve requests exceeding the peer nodes' capacities. To analytically characterize viewing demands and PCDN service capacities, we adopt a discrete time-slot model, assuming each user requests and watches one short video per time slot\footnote{This assumption is for mathematical tractability. In practice, the granularity of videos can be refined into uniform video chunks, and a time slot can be correspondingly aligned with the playback duration of a single chunk.}. Correspondingly, the concurrent service capacity of a PCDN node is defined as the maximum number of video requests it can serve within a single time slot. The central server operates based on two primary inputs, user-video recommendation sets (the specific videos intended for each user) and node-video storage information (the videos cached at each PCDN node and the concurrent service capacities of the PCDN nodes). Based on these inputs, the server generates a personalized playlist for each user that not only predefines the playback order but also assigns the designated delivery node for each video.

Based on the established system model, we formulate the Optimal Video Ordering and Transmission Scheduling (OVOTS) problem, an Integer Linear Program (ILP) aiming to minimize the overall transmission cost by determining the optimal playback sequence for each user and the corresponding node assignment for each video. To overcome the inherent computational complexity of standard ILP solvers, we rigorously prove that the OVOTS problem can be equivalently reduced to an auxiliary Minimum Cost Maximum Flow (MCMF) problem~\cite{chen2025maximum} by mapping the concurrent service capacities to link capacities within the flow network. Furthermore, leveraging König’s Edge Coloring Theorem~\cite{schrijver1998bipartite} on bipartite graphs, we demonstrate that any optimal solution to the MCMF problem can be seamlessly transformed into a optimal solution for the original OVOTS problem without loss of optimality. Guided by this theoretical foundation, we propose the Minimum-cost Maximum-flow with Edge Coloring (MMEC) algorithm, which operates in two strictly sequential phases. First, a scheduling phase solves the MCMF instance via the Successive Shortest Path algorithm based on Bellman-Ford~\cite{parimala2021bellman} to establish optimal user-node pairings; subsequently, an ordering phase employs a Kempe chain coloring algorithm~\cite{robertson1997four} on the established pairings to map videos into specific time slots (i.e., generating the final playback sequence). By exploiting the specific mathematical structure of our formulation, this decoupled approach guarantees a globally optimal polynomial-time solution, offering significant scalability advantages over traditional exponential-time ILP solvers.

Finally, to evaluate the effectiveness of the proposed joint optimization framework, we conduct extensive simulation experiments using synthetic traces that follow a Zipf popularity distribution to faithfully emulate real-world short video access patterns. Specifically, we benchmark our core algorithm, MMEC, against three representative baselines: Random Ordering with Random Transmission Scheduling (RORS), Random Ordering with Optimal Transmission Scheduling (ROOS), and Simulated Annealing Optimization (SAO). The empirical results demonstrate that MMEC consistently achieves superior performance, reducing the overall transmission cost by up to 67\%, 36\%, and 36\% compared to RORS, ROOS, and SAO, respectively. Furthermore, our analysis reveals that the proposed algorithm exhibits strong robustness to variations in critical system parameters, such as user scale, video library size, and the number of peer nodes. These findings underscore the substantial potential of joint video ordering and transmission scheduling as a critical paradigm for optimizing short video delivery cost in PCDN architectures.

The remainder of this paper is organized as follows. Section~\ref{background} introduces the background of PCDN architectures and motivates our study by analyzing the flash-crowd access patterns of short videos. Section~\ref{system_model} establishes the system model and formulates the OVOTS problem. Section~\ref{proposed_algorithm} provides a detailed theoretical analysis, including the problem reduction to the auxiliary MCMF formulation, and introduces the MMEC algorithm. Section~\ref{performance_evaluation} presents the experimental setup and evaluates algorithmic performance through numerical simulations. Finally, Section~\ref{conclusion} concludes the paper and outlines potential directions for future research.

\section{Background and Motivation}
\label{background}
\subsection{Peer-to-Peer Content Delivery Networks}

PCDNs have emerged as a highly promising architecture to alleviate overall video transmission costs~\cite{wei2024pscheduler,wei2025qoe}. With the explosive growth of short video traffic~\cite{huang2023digital,huang2024digital,shang2025large}, the enormous hardware, network bandwidth, and maintenance costs associated with traditional CDNs~\cite{ghabashneh2020exploring,zhang2014unreeling} have become an unsustainable financial burden for content providers~\cite{zhang2024enhancing,wei2023multipath}. Unlike traditional CDNs that depend on expensive commercial data centers, the core insight of a PCDN is to harvest fragmented and underutilized resources distributed at the network edge~\cite{farahani2022hybrid,jiang2009efficient}. These edge nodes primarily consist of idle computing resources leased from third-party providers, such as smart home devices, household routers, and set-top boxes. As a result, PCDNs offer significant advantages in cost efficiency.

Despite this cost-effectiveness, the resource capacity at each peer node is severely limited, posing new challenges for short video transmission. Specifically, the limitations of a single peer node are mainly reflected in two dimensions: concurrent service capacity and storage capacity. According to recent large-scale measurements in commercial PCDN systems~\cite{zhang2024enhancing,wang2024twist}, the vast majority of peer nodes exhibit remarkably low service speeds. For instance, over 80\% of the peer nodes provide an average serving speed of less than 5 Mbps, and fewer than 1\% can exceed 8 Mbps~\cite{zhang2024enhancing}. In stark contrast, dedicated CDN servers easily provide throughputs in the tens to hundreds of Gbps~\cite{nygren2010akamai}. Furthermore, unlike centralized CDN servers equipped with massive enterprise-grade disk arrays (typically in the terabyte or petabyte range)~\cite{ghabashneh2020exploring,zhang2014unreeling,nygren2010akamai}, individual peer nodes must rely entirely on the residual storage of host devices, which typically constrains their cache capacity to the order of gigabytes~\cite{liu2011novasky}.

Consequently, these fundamental dual limitations severely impede effective video transmission optimization, particularly when confronting massive traffic spikes. First, the strict concurrent bandwidth limitation implies that a single peer node can only serve a very small number of simultaneous requests. When highly popular short videos trigger a flash crowd, the limited concurrent service capacity of these peer nodes is instantly saturated, leading to severe congestion. Second, the strict storage limit dictates that a single peer node can only cache a minute fraction of the entire video library. As a result, it is physically impossible for peer nodes to maintain a stable and comprehensive working set of popular content. This inherent scarcity leads to persistently high cache miss rates, especially when user demand shifts rapidly. Whenever a requested video is absent from the local cache, the transmission must bypass the peer network and fall back to the expensive CDN nodes. This forced fallback directly nullifies the traffic offloading efficiency, ultimately driving up the overall transmission cost that the PCDN architecture was originally designed to minimize.
\begin{figure*}  
  \includegraphics[width=\textwidth]{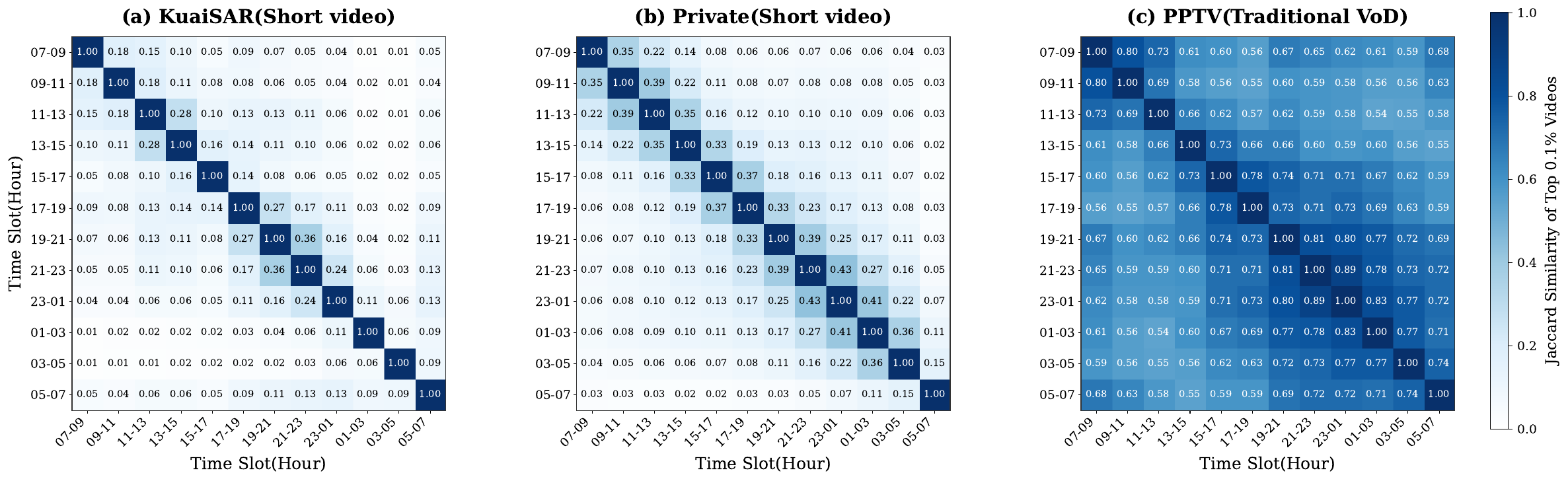}
    \caption{Jaccard similarity of the top 0.1\% most requested videos across 2-hour time slots for short video and traditional VoD datasets, illustrating the flash-crowd request pattern and rapid popularity decay of short videos compared to traditional VoD.}
    \label{fig:jaccard_heatmap_oneday}  
\end{figure*}

\begin{figure}  
  \includegraphics[width=\columnwidth]{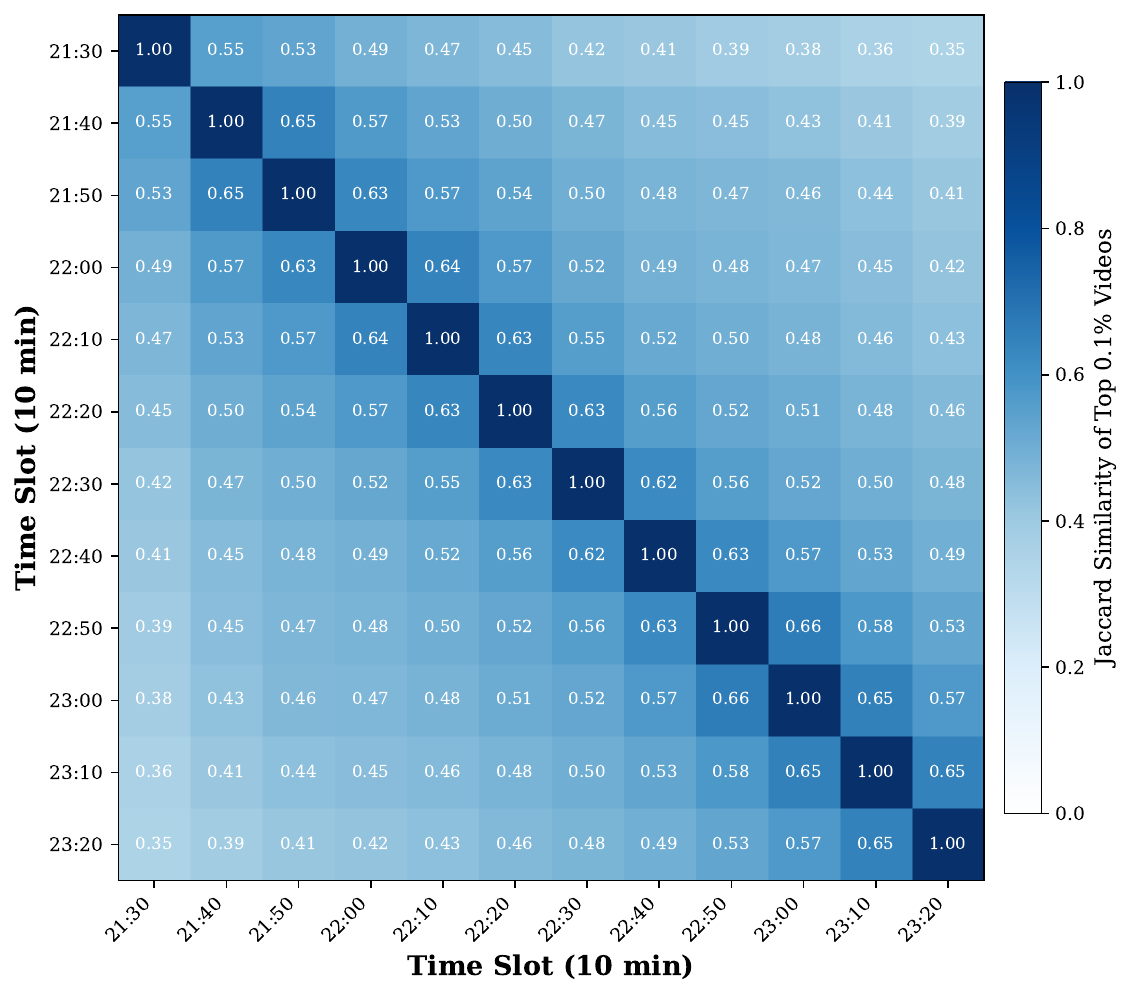}
    \caption{Jaccard similarity of the top 0.1\% most requested videos across fine-grained 10-minute time slots during the evening peak period (21:30–23:20). The similarity drops below 0.5 within just 30 minutes, demonstrating the short video flash crowds.}
    \label{fig:jaccard_peak_10min}  
\end{figure}

\begin{table}[htbp]
\caption{Basic Statistics of the 24-hour Datasets for measurement}
\label{tab:dataset_statistics}
\begin{tabular*}{\tblwidth}{@{}LLLLL@{}}
\toprule
Dataset &  Total Records & Unique Users & Unique Videos\\  

\midrule

KuaiSAR  & 770,387 & 8,398 & 379,055  \\
Private  & 24,222,188 & 255,253 & 3,198,137 \\
PPTV  & 17,873,041 & 2,712,701 & 116,252 \\
\bottomrule
\end{tabular*}
\end{table}

\subsection{The Flash-Crowd Pattern of Short Videos}
Driven by real-time recommendation engines and swipe-based continuous playback interfaces, short video platforms exhibit a pronounced flash-crowd access pattern compared to traditional Video-on-Demand (VoD) systems like YouTube~\cite{youtube} and Netflix~\cite{netflix}. Specifically, a short video often experiences massive concurrent requests within a brief time window, followed by a rapid decline in popularity. To quantitatively illustrate this characteristic, we conduct a data-driven measurement study using three distinct datasets. KuaiSAR~\cite{sun2023kuaisar} is a publicly available short video dataset collected from Kwai~\cite{kwai}; the Private dataset is a short video trace provided by a mainstream content provider; and PPTV is a traditional VoD trace from PPTV~\cite{pptv}. The fundamental statistics of these 24-hour datasets are summarized in Table~\ref{tab:dataset_statistics}.

To capture the intense temporal locality, we first divide the 24-hour period into 2-hour intervals and extract the top 0.1\% most frequently requested videos within each interval. We then compute the Jaccard similarity~\cite{bag2019efficient} of these popular video sets across different time windows. The resulting similarity matrices, presented as heatmaps in Fig.~\ref{fig:jaccard_heatmap_oneday}, reveal a striking contrast in request patterns between short video services and traditional VoD services. The Jaccard similarity in the short video datasets drops precipitously even between adjacent time slots. Conversely, in the traditional VoD dataset, the similarity remains remarkably high across the entire day. 
To further expose the microscopic severity of this flash-crowd phenomenon, we zoom in on the evening peak period (e.g., 21:30 - 23:20) of the Private dataset and evaluate the Jaccard similarity at a fine-grained 10-minute interval (as shown in Fig.~\ref{fig:jaccard_peak_10min}). Strikingly, the similarity drops below 0.5 within just 20 to 30 minutes. This empirical observation confirms that the lifespan of a short video flash crowd is incredibly short.

The combination of flash-crowd access patterns and strict storage limits renders traditional cache optimization highly ineffective. Constantly updating caches to track rapidly shifting trends induces severe cache thrashing, quickly exhausting the scarce bandwidth of peer nodes and negating the PCDN's economic benefits. Consequently, purely spatial routing or reactive caching cannot prevent massive traffic from overflowing to the expensive CDN nodes. Instead, we are motivated to exploit a unique degree of freedom inherent to short video platforms, the server-driven recommendation playlist. By intelligently permuting the temporal playback sequences of users, we intervene directly in the time domain to actively shape the demand curve, naturally smoothing out concurrent traffic peaks.

\section{System model and problem formulation}
\label{system_model}
In this section, we present the theoretical framework for our joint optimization approach. We first detail the application scenario, illustrating how optimizing playback sequences minimizes the overall transmission cost. We then formalize the system model and formulate the Optimal Video Ordering and Transmission Scheduling (OVOTS) problem as an Integer Linear Program (ILP). Finally, to overcome the inherent exponential complexity of standard ILPs, we introduce a node-splitting transformation that reduces OVOTS into a computationally tractable auxiliary formulation.

\subsection{Application scenario}

The short video streaming system studied in this paper is optimized through a joint framework of video ordering and transmission scheduling, which aims to minimize overall transmission cost by utilizing low-cost resources to serve users' video watching demands as much as possible. As shown in Fig.~\ref{fig:system_model}, the system architecture consists of a control plane and a data plane, comprising three main entities, short video servers, short video users, and PCDN nodes.

The short video server in the control plane acts as the central controller. Given user-video recommendation information and node-video storage information, the short video server executes the MMEC algorithm to generate personalized playlists for each user and determine which video to download from which storage node. Specifically, the user-video recommendation information records the video IDs that each user intends to watch, while the node-video storage information records the video IDs cached on each PCDN node and the concurrent service capacities of each PCDN node. Both are typically external inputs determined by upstream recommendation and caching algorithms, which are outside the scope of this paper.

Users download and watch short videos sequentially according to the playlist generated by the short video server. We assume that the transmission delay for distributing the playlist is negligible, as we focus on minimizing the video transmission cost in this paper. Furthermore, we assume all users consume videos at a uniform pace. The viewing process is divided into consecutive time slots (e.g., $t=1,2,3$ in Fig.~\ref{fig:system_model}). In each slot, a user downloads the required video from a pre-assigned PCDN node, watches it, and seamlessly switches to the next video via screen swiping at the end of the slot.

PCDN nodes in the data plane are categorized based on their storage capacity, concurrent service capability, and cost. The first category comprises peer nodes (e.g., $n_1, n_2$), such as home routers or set-top boxes, typically leased from third-party providers. These nodes have strictly limited storage and concurrent serving capacities, caching only a small subset of videos and serving a limited number of users concurrently per slot. The second category is the robust CDN node (e.g., $n_3$), which is assumed to possess virtually unlimited storage and bandwidth, capable of serving all users concurrently. However, the transmission cost incurred by the CDN node is significantly higher than that of peer nodes.
\begin{figure}
    \centering
    \includegraphics[width=\columnwidth]{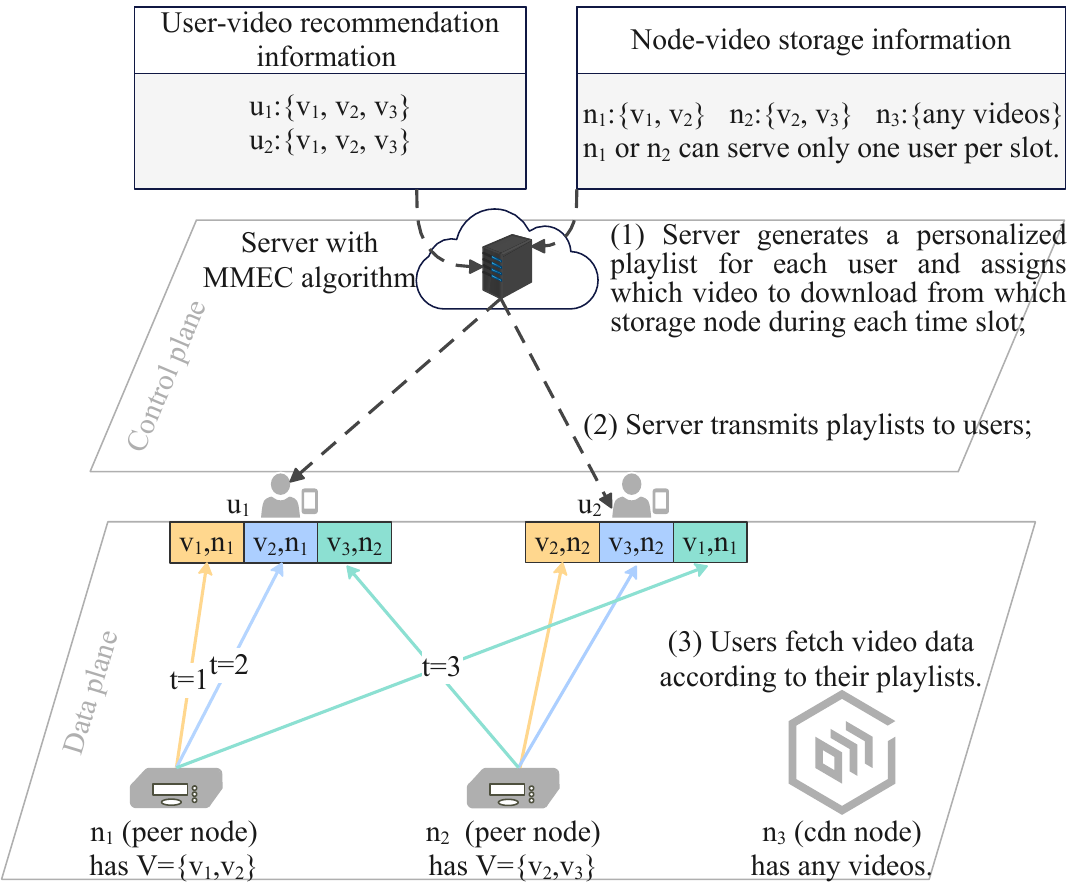}
    \caption{System architecture and motivating example of the proposed joint optimization framework.}
    \label{fig:system_model}
\end{figure}
\begin{figure}
    \centering
    \includegraphics[width=\linewidth]{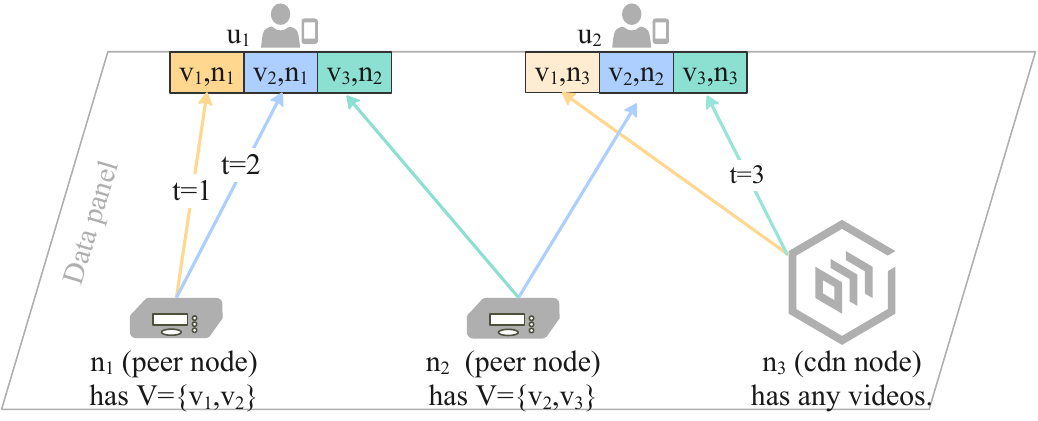}
    \caption{An example of capacity collisions under identical chronological sequences.}
    \label{fig:example}
\end{figure}

To illustrate how our joint optimization reduces service costs, consider the motivating example in Fig.~\ref{fig:system_model} and Fig.~\ref{fig:example}. Users $u_1$ and $u_2$ both intend to watch the same set of recommended videos $\{v_1, v_2, v_3\}$. If the server assigns the identical chronological sequence $<v_1, v_2, v_3>$ to both users (as shown in Fig.~\ref{fig:example}), a capacity collision occurs: at $t=1$, both attempt to fetch $v_1$, while $n_1$ can only serve one user concurrently and $n_2$ does not cache $v_1$. As a result, the server is forced to redirect $u_2$'s request to the expensive CDN node $n_3$. A similar collision also occurs at $t=3$. In contrast, by employing our proposed optimizing approach, the server reorders $u_2$'s playlist into an optimized sequence $<v_2, v_3, v_1>$ as shown in Fig.~\ref{fig:system_model}. This strategic staggering ensures that all video requests across all time slots are absorbed by the low-cost peer nodes $n_1$ and $n_2$, effectively eliminating the need for the expensive CDN node and minimizing the overall transmission cost.
  
Therefore, given such a short-video streaming system, our objective is to find the optimal playlist for each user and assign a node to each video to minimize total service cost while satisfying the storage and capacity constraints of the PCDN nodes.

\subsection{System description}

To formally model the system, we first define the sets characterizing the users and the video content. Let $\mathcal{U} = \{1, 2, \dots, U\}$ denote the set of users in the system, where $U$ represents the total number of users. The content library is represented by a set of short video files $\mathcal{V} = \{1, 2, \dots, V\}$. For ease of analysis, we assume each video has a uniform size. Furthermore, the viewing demands are captured by the user-video recommendation information, denoted by a set $\mathcal{R} = \{R_u \mid u \in \mathcal{U}\}$. Here, each element $R_u$ is a subset of $\mathcal{V}$ representing the specific videos that user $u$ intends to watch.

Next, we characterize the PCDN nodes within the data plane. We denote the set of all available nodes by $\mathcal{N} = \{0,1, 2, \dots, N\}$. To distinguish their physical properties and service costs, we define nodes $0$ through $N-1$ as low-cost peer nodes, while the last node acts as the robust, high-cost CDN node. The video caching status across the network is defined by the node-video storage information set $\mathcal{S} = \{S_n \mid n \in \mathcal{N}\}$, where each element $S_n$ specifies the subset of videos cached at node $n$. Specifically, the CDN node is assumed to store the entire video library, meaning $S_N = \mathcal{V}$. In contrast, the peer nodes have limited storage capacities and can only cache a small portion of the video data.

The operational constraints of these nodes are determined by their respective service capabilities and transmission costs. Let $c_n$ denote the transmission cost incurred when node $n$ delivers a video to a user during a single time slot. As established, the cost of the CDN node is significantly higher than that of any peer node, such that $c_N \gg c_n$ for all $n < N$. Additionally, each node $n$ is restricted by a maximum concurrent capacity $d_n$, which represents the maximum number of users it can serve simultaneously within a single time slot. Peer nodes typically have a very restricted capacity $d_n$, whereas the CDN node is assumed to possess sufficient capacity to serve all users in $\mathcal{U}$ concurrently if required, $d_N=|\mathcal{U}|$.

Finally, we model the temporal dynamics of the short video streaming system over a finite time horizon, which is partitioned into a set of discrete time slots $\mathcal{T} = \{1, 2, \dots, T\}$. Without loss of generality, to ensure the feasibility of mapping every recommended video to a distinct time slot, we assume the size of the recommendation set for each user equals the total number of time slots in the horizon, i.e., $|R_u| = T, \forall u \in \mathcal{U}$. The overall process is divided into a scheduling phase and a streaming phase. It is assumed that the scheduling phase occurs within the control plane prior to the first time slot. During this phase, the short video server processes the recommendation information $\mathcal{R}$, the storage information $\mathcal{S}$, and concurrent capacity $d_n$ to generate personalized playlists and node assignments globally. Subsequently, the actual video streaming occurs from $t=1$ to $T$. Throughout these time slots, users download and consume their assigned videos sequentially and at a uniform pace, switching to the next video at the end of each slot as dictated by the pre-computed schedule.

\begin{table}[htbp]
\centering
\caption{Symbols used.}
\label{tab:symbols}
\begin{tabular*}{\tblwidth}{@{}LL@{}}
\toprule
\textbf{Symbol} & \textbf{Description} \\ \midrule
$\mathcal{U}$   & Set of short video users\\
$\mathcal{V}$   & Set of short videos\\
$\mathcal{R}$   & User-video recommendation information\\
$R_u$           & Set of videos that user $u$ intends to watch\\
$\mathcal{N}$   & Set of PCDN nodes\\
$\mathcal{S}$   & Node-video storage information\\
$\mathcal{T}$   & Set of discrete time slots \\
$S_n$           & Set of videos stored at PCDN node $n$ \\
$c_n$           & \makecell[l]{Transmission cost of PCDN node $n$ \\ per video per time slot } \\
$d_n$           & \makecell[l]{Maximum concurrent transmission capacity\\ of node $n$} \\

\bottomrule
\end{tabular*}
\end{table}

\subsection{Problem formulation}

In this section, we formulate the Optimal Video Ordering and Transmission Scheduling (OVOTS) problem, which aims to jointly determine the optimal playback sequence for each user and assign a target node to serve each video request, thereby minimizing the total delivery cost. 

To mathematically formulate this problem, we introduce a binary decision variable $x_{u,v,n,t} \in \{0, 1\}$. Specifically, $x_{u,v,n,t} = 1$ indicates that the central server schedules video $v$ at the $t$-th position of user $u$'s playlist and instructs the user to fetch this video from node $n$ during time slot $t$, and $x_{u,v,n,t} = 0$ otherwise. The decision variable $x_{u,v,n,t}$ is primarily subject to the users' viewing demands as well as the storage and concurrent service capacities of the PCDN nodes.

To ensure the validity of the generated playlists, the server can schedule a video into a user's playlist only if it is explicitly included in that user's recommendation set. This requirement imposes the following bounding constraint.
\begin{equation}
    x_{u,v,n,t} \le \mathbb{I}(v\in R_u), \quad \forall u\in \mathcal{U}, v \in \mathcal{V}, n \in \mathcal{N}, t \in \mathcal{T}
    \label{constraint_1}
\end{equation}
where $\mathbb{I}(\cdot)$ is an indicator function that returns $1$ if the condition holds and $0$ otherwise. Based on the uniform playback pace assumption, each user consumes exactly one short video per time slot. Consequently, for any user $u$ at any time slot $t$, exactly one video must be fetched from exactly one assigned node, formulated as Eq.~\eqref{constraint_2}.
\begin{equation}
    \sum_{v\in \mathcal{V}} \sum_{n\in \mathcal{N}} x_{u,v,n,t} = 1, \quad \forall u\in \mathcal{U}, t\in \mathcal{T}
    \label{constraint_2}
\end{equation}
Furthermore, to strictly adhere to the user's recommendation set without any omission or duplication, every video $v$ within user $u$'s recommendation set $R_u$ must be scheduled and downloaded exactly once across the entire time horizon $\mathcal{T}$. This is mathematically formulated as Eq.~\eqref{constraint_3}.
\begin{equation}
    \sum_{t\in \mathcal{T}} \sum_{n\in \mathcal{N}} x_{u,v,n,t} = \mathbb{I}(v\in R_u), \quad \forall u\in \mathcal{U}, v\in \mathcal{V}
    \label{constraint_3}
\end{equation}

From the perspective of video availability, a node $n$ can be assigned to serve video $v$ only if the video is physically cached in its local storage $S_n$, which can be formulated  as follows.
\begin{equation}
    x_{u,v,n,t} \le \mathbb{I}(v\in S_{n}), \quad \forall u\in \mathcal{U}, v \in \mathcal{V}, n \in \mathcal{N}, t \in \mathcal{T}
    \label{constraint_4}    
\end{equation}
Crucially, to strictly adhere to the concurrent service limits, the total number of concurrent users served by any PCDN node $n$ during any single time slot $t$ cannot exceed its maximum concurrent service capacity $d_n$. Regardless of which specific videos are being requested, this capacity constraint is strictly guaranteed by Eq.~\eqref{constraint_5}
\begin{equation}
    \sum_{u\in \mathcal{U}} \sum_{v\in \mathcal{V}} x_{u,v,n,t} \le d_n, \quad \forall n\in \mathcal{N}, t\in \mathcal{T}
    \label{constraint_5}    
\end{equation}

Given a valid joint ordering and scheduling policy defined by the matrix $\boldsymbol{X} = \{x_{u,v,n,t}\}$, the global video delivery cost is the sum of the transmission costs incurred by all nodes serving the requests across all time slots. By denoting this total cost as $C_{\boldsymbol{X}} = \sum_{u\in \mathcal{U}} \sum_{v \in \mathcal{V}} \sum_{n\in \mathcal{N}} \sum_{t\in \mathcal{T}} c_n x_{u,v,n,t}$, the OVOTS problem can be ultimately formulated as the following Integer Linear Program (ILP) problem.
\begin{equation}
\begin{aligned}
    \mathcal{P}: \quad &\min_{\boldsymbol{X}} \quad C_{\boldsymbol{X}} \\
    s.t.\quad
    &x_{u,v,n,t}\in\{0,1\}, \quad \forall u \in \mathcal{U}, v\in \mathcal{V}, n\in \mathcal{N}, t\in \mathcal{T} \\ 
    &\text{Eqs.~\eqref{constraint_1}} - \text{~\eqref{constraint_5}} 
\end{aligned}
\end{equation}
\textbf{Problem Equivalent Transformation:} Solving the problem $\mathcal{P}$ as formulated above directly is computationally intractable due to the inherent exponential complexity of standard ILP solvers~\cite{Gurobi,OR-tools}. To facilitate the design of a polynomial-time solution, we transform the optimization problem $\mathcal{P}$ into an equivalent auxiliary problem $\mathcal{P}'$ by applying a node splitting transformation on the physical PCDN nodes. Specifically, any physical node $n \in \mathcal{N}$ characterized by a maximum concurrent capacity $d_n$, a storage profile $S_n$, and a unit transmission cost $c_n$ is expanded into a collection of $d_n$ independent virtual nodes.  
Each virtual node $n' \in \mathcal{N}'$ inherits the identical storage profile $S_{n'} = S_n$ and transmission cost $c_{n'} = c_n$ from its original physical node $n$, but importantly, possesses a strict unit concurrent service capacity $d_{n'} = 1$. 
Let $\mathcal{N}'$ denote the expanded set of these virtual nodes, sized at $N' = \sum_{n \in \mathcal{N}} d_n$. By redefining the decision variable as $x_{u,v,n',t}$ over the virtual node set $\mathcal{N}'$, and redefining the cost function as $C_{\boldsymbol{X}} = \sum_{u\in \mathcal{U}} \sum_{v \in \mathcal{V}} \sum_{n'\in \mathcal{N}'} \sum_{t\in \mathcal{T}} c_{n'} x_{u,v,n',t}$, the problem $\mathcal{P'}$ can be formulated as follows.
\begin{equation}
    \begin{aligned}
    \mathcal{P'}: \quad &\min_{\boldsymbol{X}} \quad C_{\boldsymbol{X}} \\
    s.t.\quad
    &x_{u,v,n',t}\in\{0,1\}, \quad \forall u \in \mathcal{U}, v\in \mathcal{V}, n'\in \mathcal{N'}, t\in \mathcal{T} \\ 
     &x_{u,v,n',t} \le \mathbb{I}(v\in R_u),\forall u \in \mathcal{U}, v\in \mathcal{V}, n'\in \mathcal{N'}, t\in \mathcal{T} \\
     &\sum_{v\in \mathcal{V}} \sum_{n'\in \mathcal{N'}} x_{u,v,n',t} = 1, \quad \forall u\in \mathcal{U}, t\in \mathcal{T}\\
     &\sum_{t\in \mathcal{T}} \sum_{n'\in \mathcal{N'}} x_{u,v,n',t} = \mathbb{I}(v\in R_u), \quad \forall u\in \mathcal{U}, v\in \mathcal{V}\\
      &x_{u,v,n',t} \le \mathbb{I}(v\in S_{n'}), \forall u\in \mathcal{U}, v \in \mathcal{V}, n' \in \mathcal{N'}, t \in \mathcal{T}\\
    & \sum_{u\in \mathcal{U}} \sum_{v\in \mathcal{V}} x_{u,v,n',t} \le 1, \quad \forall n'\in \mathcal{N}', t\in \mathcal{T}\\   
\end{aligned}
\end{equation}
Problem $\mathcal{P}'$ mathematically preserves all feasible solutions and the exact optimal objective value of $\mathcal{P}$. This structural transformation effectively maps the complex multi-capacity scheduling problem into a tractable bipartite assignment structure, directly enabling the graph-theoretic reduction presented in the subsequent section.

\section{Theoretical Analysis and Proposed Algorithm }
\label{proposed_algorithm}

In this section, we propose a reduction-based algorithmic framework to efficiently solve the transformed optimization problem $\mathcal{P}'$. The core of our approach lies in mapping the integer linear program into an auxiliary Minimum Cost Maximum Flow (MCMF) problem. We theoretically prove that the optimal objective value of this MCMF instance is strictly equivalent to that of $\mathcal{P}'$. Building upon this rigorous equivalence, we introduce a polynomial-time, two-phase algorithm, Minimum-cost Maximum-flow with Edge Coloring (MMEC), that reconstructs an optimal video ordering and transmission scheduling.

\subsection{MCMF Problem Construction}
\label{MCMF_problem}

To resolve the time-coupling nature of the joint optimization problem $\mathcal{P'}$, we temporarily relax the discrete time-slot constraints and aggregate the system capacities over the entire time horizon to construct a directed flow network $\mathcal{G} = (\mathcal{K}, \mathcal{E})$. 

\begin{figure}
    \centering
    \includegraphics[width=\linewidth]{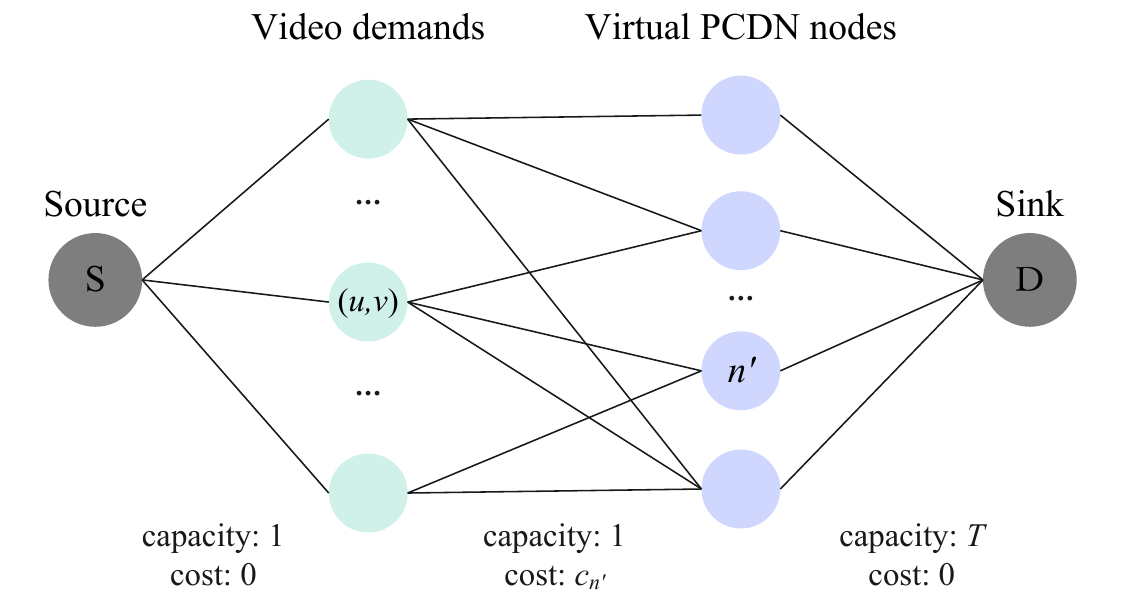}
    \caption{Auxiliary flow network representing the aggregated video delivery requests and virtual node capacities.}
    \label{fig:flow_network}
\end{figure}

As illustrated in Fig.~\ref{fig:flow_network}, the vertex set $\mathcal{K}$ comprises a source $S$, a sink $D$, a set of demand nodes $\{(u,v) \mid u \in \mathcal{U}, v \in R_u\}$ representing the video requests of each user, and the expanded virtual nodes set $\mathcal{N}'$. The directed edge set $\mathcal{E}$ is constructed based on the following physical mapping rules:

\begin{itemize}

    \item \textbf{Demand Generation Edges:} For each user $u \in \mathcal{U}$ and each explicitly recommended video $v \in R_u$, we draw an edge $S \to (u,v)$ with a capacity of $1$ and a cost of $0$. This rigorously restricts each recommended video to be downloaded exactly once by the corresponding user.
    \item \textbf{Video Transmission Edges:} An edge $(u,v) \to n'$ exists if and only if video $v$ is physically cached at the virtual node $n'$ (i.e., $v \in S_{n'}$). Each such edge is assigned a capacity of $1$ and a transmission cost of $c_{n'}$.
    \item \textbf{Node Capacity Edges:} Each virtual PCDN node $n' \in \mathcal{N}'$ is connected to the sink $D$ via an edge $n' \to D$ with a capacity of $T$ and a cost of $0$. This elegantly enforces that a virtual PCDN node, possessing a strictly unit concurrent capacity per time slot, can serve at most $T$ video requests aggregately over the entire time horizon $\mathcal{T}$.
\end{itemize}

Let $\mathcal{F}$ denote the MCMF problem defined on the network $\mathcal{G}$, which seeks to push the exact total required flow (i.e., $\sum_{u \in \mathcal{U}} |R_u|$) from the source $S$ to the sink $D$ while minimizing the overall transmission cost~\cite{ford1962flows}.

\subsection{Equivalence Analysis}
\label{Equivalence_Analysis}

We now establish the optimality of our reduction by proving the rigorous equivalence between the auxiliary MCMF problem $\mathcal{F}$ and the transformed scheduling problem $\mathcal{P}'$.

\begin{theorem}
Let $OPT_{\mathcal{P}'}$ be the minimum overall delivery cost of the scheduling problem $\mathcal{P}'$, and let $OPT_{\mathcal{F}}$ be the minimum cost of the constructed network flow problem $\mathcal{F}$. Then, $OPT_{\mathcal{P}'} = OPT_{\mathcal{F}}$.
\end{theorem}

\begin{proof}
The proof relies on establishing a bidirectional bounding relationship. We first prove $OPT_{\mathcal{F}} \le OPT_{\mathcal{P}'}$ by demonstrating that any optimal solution for problem $\mathcal{P}'$ corresponds to a feasible flow in problem $\mathcal{F}$ with the same cost. Subsequently, we prove $OPT_{\mathcal{P}'} \le OPT_{\mathcal{F}}$ by showing that any optimal flow for $\mathcal{F}$ can be losslessly decomposed into a feasible scheduling matrix for $\mathcal{P}'$.

\noindent \textbf{(1) $OPT_{\mathcal{F}} \le OPT_{\mathcal{P}'}$.} 

\textbf{Solution Construction:} Given an optimal scheduling matrix $\boldsymbol{X}^*=\{x^*_{u,v,n',t} \mid u \in \mathcal{U}, v\in \mathcal{V}, n'\in \mathcal{N}', t\in \mathcal{T}\}$ for problem $\mathcal{P}'$, we construct a set of flows $F$ on network $\mathcal{G}$. For every $x^*_{u,v,n',t} = 1$, we route one unit of flow along the path $S \to (u,v) \to n' \to D$. Thus, the total flow mapped to the network is $F = \{f(S \to (u,v) \to n' \to D) = 1 \mid x^*_{u,v,n',t} = 1\}$.

\textbf{Constraints Satisfaction \& Cost Match:} First, the flows in $F$ do not violate the unit capacity of the Demand Generation Edges $S \to (u,v)$, because the constraints of $\mathcal{P}'$ ensure each user fetches each recommended video exactly once across $\mathcal{T}$. Second, the flows do not violate the unit capacity of the Video Transmission Edges $(u,v) \to n'$, since each specific video request is served by exactly one assigned virtual node. Third, the flows do not violate the aggregated capacity $T$ of the Node Capacity Edges $n' \to D$, as the capacity constraint in $\mathcal{P}'$ dictates that each virtual node $n'$ serves at most one user per time slot, cumulatively serving at most $T$ requests over $T$ time slots. Since the flow mapping is strictly one-to-one with the decision variables, the total flow cost exactly equals the scheduling cost, i.e., $Cost(F) = C_{\boldsymbol{X}^*} = OPT_{\mathcal{P}'}$. Because $F$ is a valid feasible solution for $\mathcal{F}$, the optimal cost of $\mathcal{F}$ can be no greater than the cost of $F$. Thus, $OPT_{\mathcal{F}} \le OPT_{\mathcal{P}'}$.

\noindent \textbf{(2) $OPT_{\mathcal{P}'} \le OPT_{\mathcal{F}}$.}

\begin{figure}
    \centering    \includegraphics[width=0.7\linewidth]{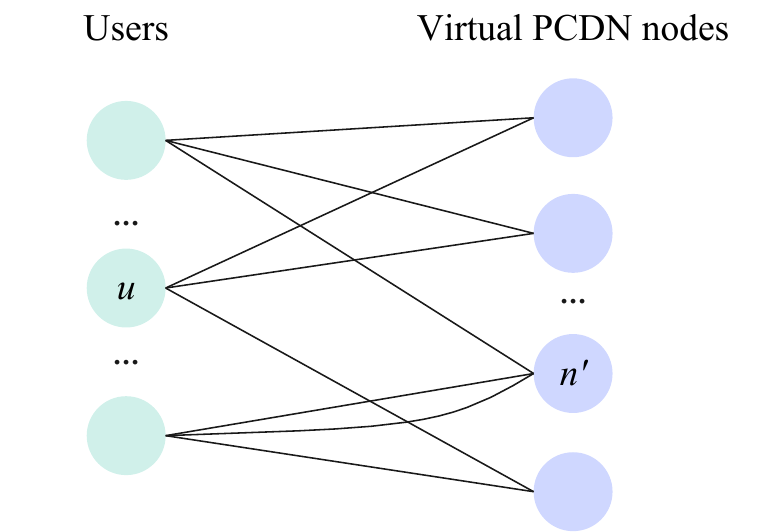}
    \caption{Bipartite multigraph $\mathcal{G_B}$ constructed from the optimal flow $F^*$. Parallel edges represent multiple distinct video requests between a user and a virtual PCDN node.}
    \label{fig:bipartite_multigraph}
\end{figure}
\textbf{Solution Construction:} Given an optimal integer flow solution $F^*=\{f^*(S \to (u,v) \to n' \to D) = 1\}$ for problem $\mathcal{F}$, we construct a bipartite multigraph $\mathcal{G_B} = (\mathcal{U}, \mathcal{N}', \mathcal{E}_B)$ as shown in Fig.~\ref{fig:bipartite_multigraph}, where $\mathcal{U}$ represents the users and $\mathcal{N}'$ represents the virtual PCDN nodes. For each unit flow $f^*(S \to (u, v) \to n' \to D) = 1$ in $F^*$, we add an edge connecting $u \in \mathcal{U}$ and $n' \in \mathcal{N}'$, labeled uniquely by the video $v$ as $e_{u,v,n'}$.

\textbf{Constraints Satisfaction \& Cost Match:} In the bipartite multigraph $\mathcal{G_B}$, the degree of each vertex $u \in \mathcal{U}$ is exactly $T$, because each user needs to download and watch $T$ videos. Furthermore, the degree of each vertex $n' \in \mathcal{N}'$ cannot exceed $T$ due to the strict flow capacity of edge $n' \to D$. Thus, the maximum degree of $\mathcal{G_B}$ is bounded by $\Delta(\mathcal{G_B}) \le T$. According to \textit{K\H{o}nig's Edge Coloring Theorem}, any bipartite multigraph with a maximum degree $\Delta$ is $\Delta$-edge-colorable. This theoretically guarantees that we can color all edges in $\mathcal{G_B}$ using a set of exactly $T$ colors, such that no two adjacent edges share the same color. 

By bijectively mapping the $T$ colors to the $T$ discrete time slots in $\mathcal{T}$, we set the decision variable $x_{u,v,n',t} = 1$ if edge $e_{u,v,n'}$ is colored with color $t$, which yields a complete scheduling matrix $\boldsymbol{X}$. K\H{o}nig's theorem rigorously ensures that no user fetches more than one video per slot, and no virtual PCDN node serves more than one user per slot, satisfying all constraints of $\mathcal{P}'$. Since the physical routing paths remain unchanged during the coloring process, the transmission cost is perfectly preserved($C_{\boldsymbol{X}} = Cost(F^*) = OPT_{\mathcal{F}}$). Therefore, $\boldsymbol{X}$ is a feasible solution for $\mathcal{P}'$, implying $OPT_{\mathcal{P}'} \le OPT_{\mathcal{F}}$.

Combining the bounds from (1) and (2), we conclusively establish $OPT_{\mathcal{P}'} = OPT_{\mathcal{F}}$.
\end{proof}

\subsection{Algorithm Description}

Guided by the theoretical foundation established in the previous subsections, we propose the Minimum-cost Maximum-flow with Edge Coloring (MMEC) algorithm. As outlined in Algorithm 1, the framework efficiently decouples the joint optimization problem into two strictly sequential phases: a flow-based transmission scheduling phase and a coloring-based video ordering phase.

The first phase (lines 3-13) focuses on determining the optimal assignment of video requests to virtual PCDN nodes to minimize the aggregate transmission cost. Initially, MMEC constructs the auxiliary network $\mathcal{G}$ and initializes a residual graph $\mathcal{G}_r$ along with an empty flow set $F$. Then, MMEC iteratively finds the minimum cost path on $\mathcal{G}_r$ until no further valid paths exist. During each iteration, the Bellman-Ford algorithm is employed to identify a path from the source $S$ to the sink $D$ with the minimum cost. Subsequently, a unit flow along this minimum cost path, denoted as $f(S \to (u,v) \to n' \to D) = 1$, is added to the set $F$. This mathematically establishes the optimal user-node pairing, representing that virtual node $n'$ is permanently assigned to serve user $u$'s request for video $v$. The residual graph $\mathcal{G}_r$ is updated accordingly by decrementing the capacities of the corresponding edges. Finally, MMEC outputs the flow set $F$ as the optimal solution for the MCMF problem.

Although Phase 1 optimally resolves the transmission scheduling (i.e., which virtual PCDN nodes serve which video demands), it does not specify the final playback order. Phase 2 (lines 14-20) resolves this temporal mapping by transforming the optimal flow $F$ into a feasible, collision-free time-slotted schedule. MMEC first constructs a bipartite multigraph $\mathcal{G_B} = (\mathcal{U}, \mathcal{N'}, \mathcal{E'})$, where the vertex sets $\mathcal{U}$ and $\mathcal{N'}$ represent the users and virtual PCDN nodes, respectively. For every unit flow in $F$ that assigns video $v$ to virtual node $n'$ for user $u$, a unique edge $e_{u,v,n'}$ is added to $\mathcal{E'}$. Subsequently, for each edge in $\mathcal{G_B}$, MMEC uses the Kempe chain Edge Coloring algorithm to identify and assign an available color $t \in \mathcal{T}$, ensuring no adjacent edges share the same color. Finally, MMEC outputs the scheduling matrix $\boldsymbol{X}$, where $x_{u,v,n',t}=1$ directly indicates that video $v$ is ordered to be fetched by user $u$ from virtual node $n'$ strictly during time slot $t$.

Finally, to generate the solution for the original problem$\mathcal{P}$, MMEC performs a mapping step in lines 21-25. It naturally projects the virtual scheduling matrix $\boldsymbol{X}'$ back to the physical domain by restoring each virtual node $n'$ to its physical parent node $n$. Consequently, it outputs the final joint scheduling matrix $\boldsymbol{X} = \{x_{u,v,n,t}\}$, directly providing the optimal solution for the original problem $\mathcal{P}$.

\begin{algorithm}
\caption{Minimum-cost Maximum-flow with Edge Coloring (MMEC)}
\begin{algorithmic}[1]
\Require User-video recommendation sets $\{R_u\}$, Physical node storage profiles $\{S_n\}$, costs $c_n$, and capacities $d_n$
\Ensure Optimal joint scheduling matrix $\boldsymbol{X}$ for problem $\mathcal{P}$

\State Expand physical nodes $\mathcal{N}$ into virtual nodes $\mathcal{N}'$ where $N' = \sum_{n \in \mathcal{N}} d_n$
\State Initialize scheduling matrix $\boldsymbol{X}' \leftarrow \emptyset$

\State \textbf{Phase 1: Flow-based Transmission Scheduling (Solving MCMF)}
\State Construct network graph $\mathcal{G}$ based on virtual nodes $\mathcal{N}'$ as defined in Sec.~\ref{MCMF_problem}
\State Initialize flow set $F \leftarrow \emptyset$ 
\State Initialize residual graph $\mathcal{G}_r \leftarrow \mathcal{G}$
\While {True} 
    \State Find a minimum cost path $S\to (u,v)\to n'\to D$ in $\mathcal{G}_r$ using Bellman-Ford
    \If{no path exists} \textbf{break} \EndIf
    \State $F \leftarrow F \cup \{f(S\to (u,v)\to n'\to D)=1\}$
    \State Update capacities in residual graph $\mathcal{G}_r$
\EndWhile

\State \textbf{Phase 2: Coloring-based Video Ordering (Solving Edge Coloring)}
\State Construct bipartite multigraph $\mathcal{G_B} = (\mathcal{U}, \mathcal{N'}, \mathcal{E'})$ based on flow set $F$
\State Initialize color set representing time slots $\mathcal{T}=\{1,2,\dots,T\}$
\For{each edge $e_{u,v,n'}$ in $\mathcal{G_B}$}
   \State Assign a color $t \in \mathcal{T}$ using Kempe chain Coloring Algorithm
   \State $\boldsymbol{X}' \leftarrow \boldsymbol{X}' \cup \{x_{u,v,n',t}=1\}$
\EndFor

\State Initialize final physical scheduling matrix $\boldsymbol{X} \leftarrow \emptyset$
\For{each $x_{u,v,n',t}=1 \in \boldsymbol{X}'$}
   \State Map virtual node $n'$ back to its physical node $n$
   \State $\boldsymbol{X} \leftarrow \boldsymbol{X} \cup \{x_{u,v,n,t}=1\}$
\EndFor
\State \Return $\boldsymbol{X}$

\end{algorithmic}
\end{algorithm}

\subsection{Complexity Analysis}
To provide a comprehensive theoretical foundation for the proposed MMEC algorithm, we rigorously evaluate its computational efficiency. 
The overall complexity is determined by the two sequential phases. 

The computational complexity of Phase 1 is governed by the scale of the auxiliary flow network $\mathcal{G} = (\mathcal{K}, \mathcal{E})$. Specifically, the vertex set size is bounded by $|\mathcal{K}| = UT + N' + 2 = O(UT + N')$, comprising the video demand nodes, virtual storage nodes, source, and sink. The edge set size is bounded by the sum of demand generation edges, video transmission edges, and node capacity edges. Since the number of transmission edges $|\mathcal{E}_{trans}|$ is strictly bounded by the product of demand nodes and virtual nodes ($|\mathcal{E}_{trans}| \le UTN'$), the total edge count is $|\mathcal{E}| = O(UTN')$.
The Successive Shortest Path process utilizing the Bellman-Ford algorithm is performed exactly $UT$ times (once for each unit of required flow). Since a single Bellman-Ford execution takes $O(|\mathcal{K}| \cdot |\mathcal{E}|) = O(U^2T^2N' + UTN'^2)$ time, the cumulative complexity for Phase 1 reaches $O(U^3T^3N' + U^2T^2N'^2)$.
Phase 2 involves coloring the edges of the constructed bipartite multigraph $\mathcal{G}_B = (\mathcal{U}, \mathcal{N}', \mathcal{E}')$. The exact number of edges in this graph is $|\mathcal{E}'| = UT$. In the worst-case scenario, the Kempe chain coloring algorithm resolves a single edge conflict in $O(U + N')$ time by searching for an alternating path. Executing this for all $UT$ edges yields a total complexity of $O(U^2T + UTN')$ for Phase 2.

Combining all operational phases, the global computational complexity of the MMEC algorithm is strictly dominated by Phase 1 (solving the MCMF problem). Consequently, the total time complexity is bounded by $O(U^3T^3N' + U^2T^2N'^2)$.

\section{Performance Evaluation}
\label{performance_evaluation}

In this section, we conduct a comprehensive evaluation of the proposed MMEC algorithm through a series of numerical simulations. We first define the experimental environment, data generation models, and parameter settings. Subsequently, we introduce the baseline algorithms and analyze the performance gains of our joint video ordering and transmission scheduling framework under various system parameter settings.

\subsection{Experimental Setup}

\textbf{Problem instance generation.} To ensure clarity and tractability in our core performance evaluation, we construct a homogeneous PCDN, where all peer nodes are characterized by uniform concurrent serving capacities (denoted by $d$), identical storage limits (denoted by $s$), and equivalent unit transmission costs\footnote{Supplementary experiments demonstrating the robust applicability of our framework under heterogeneous peer node architectures are thoroughly detailed in Appendix~\ref{heterogeneous_peer_node}.}. Specifically, the unit transmission cost for any regular peer node is normalized to $c_n = 1$ for all $n \in \{1, \dots, N-1\}$, whereas the unit cost of fetching from the central CDN node is penalized at $c_N = 5$ to reflect the expensive backbone bandwidth consumption. 

Within this PCDN, the global delivery cost is heavily dictated by the intricate coupling between the physical video caching and the user demand patterns. To generate realistic user-video requests, we assume the popularity of the entire video library follows a Zipf distribution with a skewness parameter $\alpha$, strictly mirroring the highly concentrated viewing behaviors observed in real-world short video platforms. Based on this distribution, we independently sample exactly $T$ videos for each user $u$ to construct their personalized viewing demand set $R_u$. 

Concurrently, we implement a popularity-sorted cyclic placement strategy to generate the set of videos cached in each peer node, $\mathcal{S}$. Specifically, we rank the video library in descending order of popularity and sequentially distribute these videos to the peer nodes in a cyclic manner. This structural initialization achieves two critical objectives: first, it aims to guarantee that every video is explicitly cached by at least one low-cost peer node; second, it distributes the highly popular videos across the peer nodes, fundamentally preventing concurrent capacity bottlenecks at any single peer node.\footnote{Alternative peer node storage initialization strategies are also evaluated in Appendix~\ref{storage_strategy} to further validate the system's robustness.}

\textbf{Methodology and parameter settings.} To systematically investigate the sensitivity of the overall transmission cost to various system scales, resource capacities, and user demand patterns, we employ a controlled variable methodology. As summarized in Table~\ref{tab:parameter_settings}, we sequentially vary a single target parameter across a specified range while strictly maintaining all other system variables at their default values. 

\begin{table}[htbp]
\centering
\caption{Parameter Settings for Controlled Variable Experiments}
\label{tab:parameter_settings}
\begin{tabular*}{\tblwidth}{@{}LLLL@{}}
\toprule
\textbf{Parameter Name} & \textbf{Default} & \textbf{Range} & \textbf{Step} \\
\midrule
Number of users ($U$) & 100 & 50--140 & 10 \\
Number of peer nodes ($N$) & 50 & 10--100 & 10 \\
Number of videos ($V$) & 300 & 100--500 & 50 \\
Number of time slots ($T$) & 10 & 6--20 & 2 \\
\makecell[l]{Storage capacity per \\ peer node ($s$)} & 6 & 2--10 & 1 \\
\makecell[l]{Concurrent service capacity \\per peer node ($d$) }& 2 & 1--5 & 1 \\
Popularity skewness ($\alpha$) & 0.6 & 0.1--1.0 & 0.1 \\
\bottomrule
\end{tabular*}
\end{table}

\subsection{Comparative Algorithms}
To comprehensively evaluate the efficacy of our joint optimization framework and the proposed MMEC algorithm, we benchmark its performance against three representative baselines. These algorithms are carefully selected to represent different degrees of system optimization, serving as both performance benchmarks and ablation studies.

\textbf{Random video Ordering and Random Transmission Scheduling (RORS).} This serves as a naive, fully unoptimized baseline. In RORS, the server generates a completely random playback sequence for each user based on their recommendation set $R_u$. For each required video, the server randomly assigns a service node from the subset of peer nodes that currently cache the content and possess available concurrent capacity.

\textbf{Random video Ordering and Optimal Transmission Scheduling (ROOS).} This algorithm acts as a critical ablation study baseline. It isolates the transmission scheduling optimization while discarding the video ordering optimization. In ROOS, the playback sequence remains randomly generated; however, based on this fixed sequence, the server computes the optimal transmission scheduling to minimize the overall transmission cost. Comparing MMEC with ROOS directly quantifies the immense performance gain unlocked by intelligently permuting the user playback sequence.

\textbf{Simulated Annealing Optimization (SAO).} This baseline represents a classical meta-heuristic approach applied directly to solve our original joint Integer Linear Program (ILP) formulation. It begins with a greedy initial solution where the playback sequences are randomized, and each video request is assigned to the lowest-cost available node. The algorithm then iteratively explores the solution space by randomly swapping video time slots and reassigning service nodes to escape local optima. The SAO configuration is rigorously tuned, utilizing an initial temperature $T_{init} = 1000$, a cooling factor $\gamma = 0.95$, a maximum of $I=1000$ iterations, and $L=20$ inner perturbations per temperature step.

\begin{figure*}
    \centering
    \subfloat[Overall cost vs. number of users $U$.]{
        \label{fig:performance_comparison_num_users}
        \includegraphics[width=0.30\textwidth]{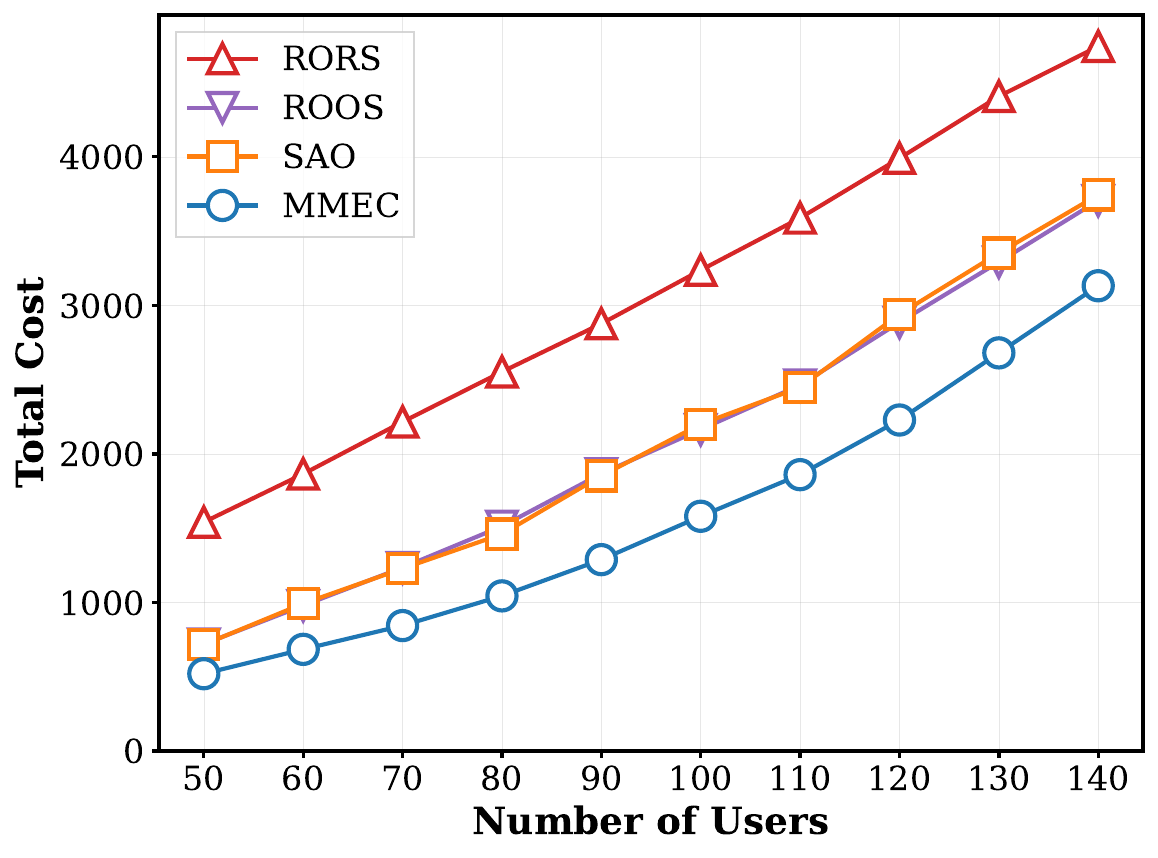}
    }
    \hfill
    \subfloat[Overall cost vs. video library size $V$.]{
        \label{fig:performance_comparison_num_files}
        \includegraphics[width=0.30\textwidth]{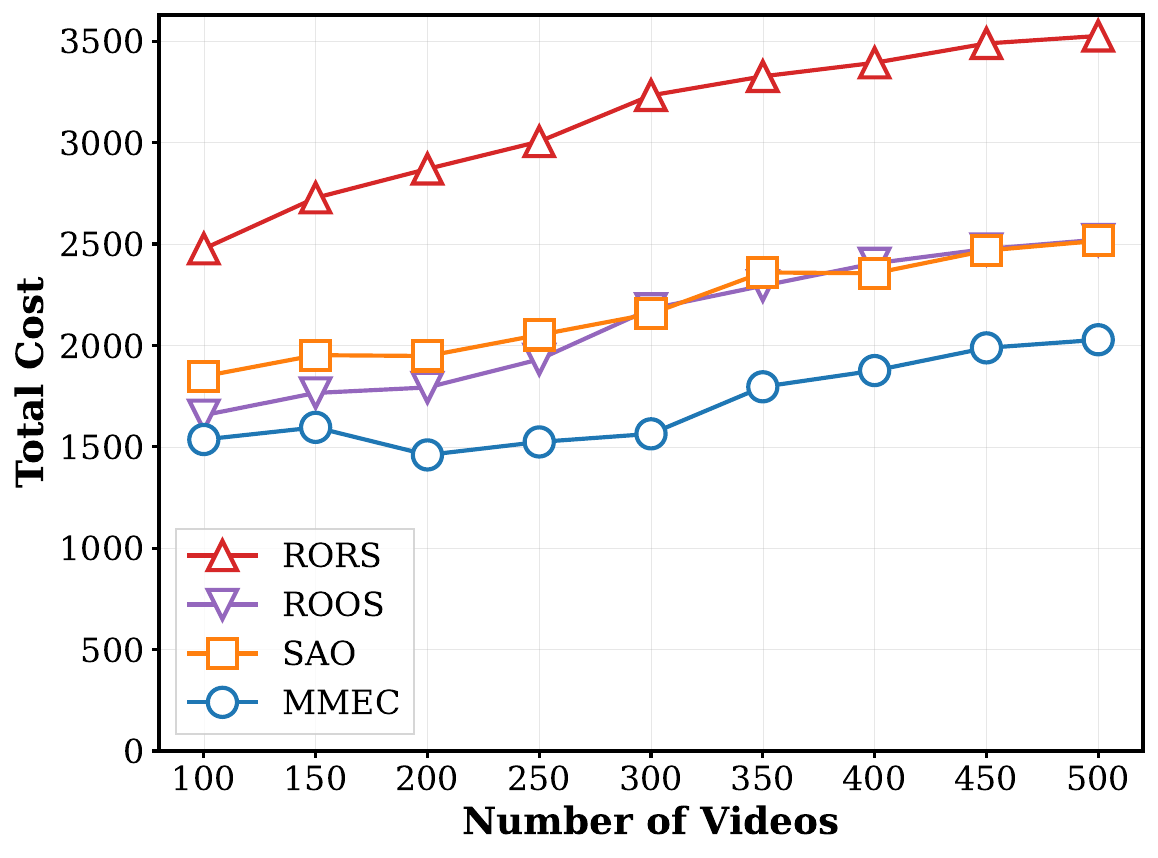}
    }
    \hfill
    \subfloat[Overall cost vs. number of peer nodes $N$.]{
        \label{fig:performance_comparison_num_servers}
        \includegraphics[width=0.30\textwidth]{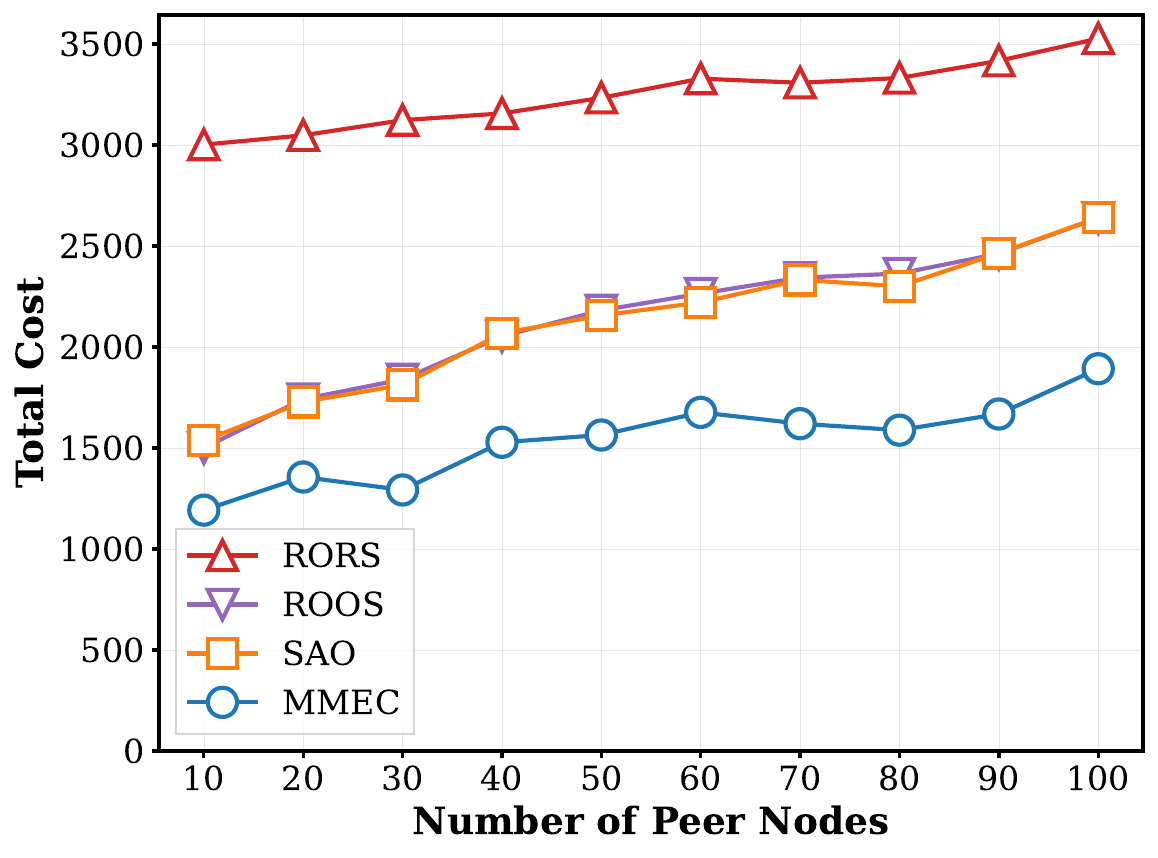}
    } \\
    
    \subfloat[Overall cost vs. number of time slots $T$.]{
        \label{fig:performance_comparison_num_slots}
        \includegraphics[width=0.30\textwidth]{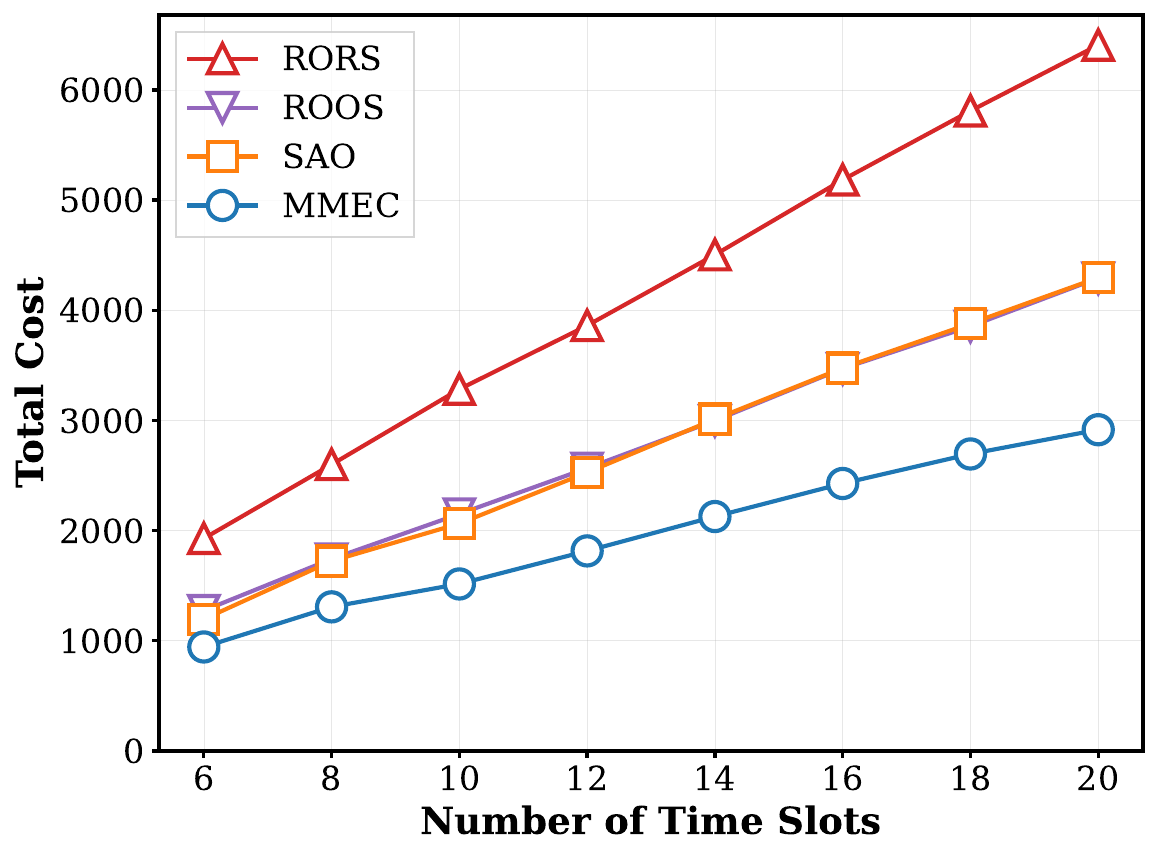}
        }
    \hfill
    \subfloat[Overall cost vs. concurrent service capacity per peer node.]{
        \label{fig:performance_comparison_mcc}
        \includegraphics[width=0.30\textwidth]{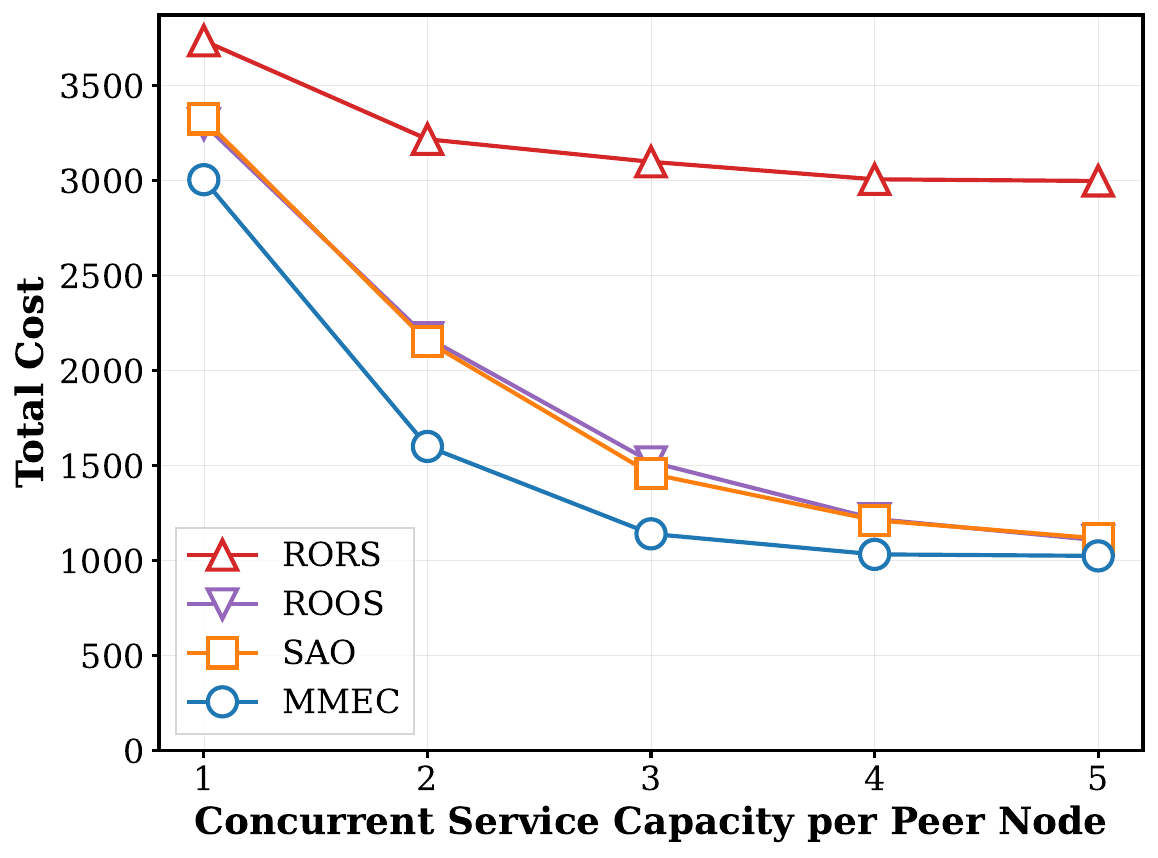}
        }
        \hfill
    \subfloat[Overall cost vs. storage capacity per peer node.]{
        \label{fig:performance_comparison_storage}
        \includegraphics[width=0.30\textwidth]{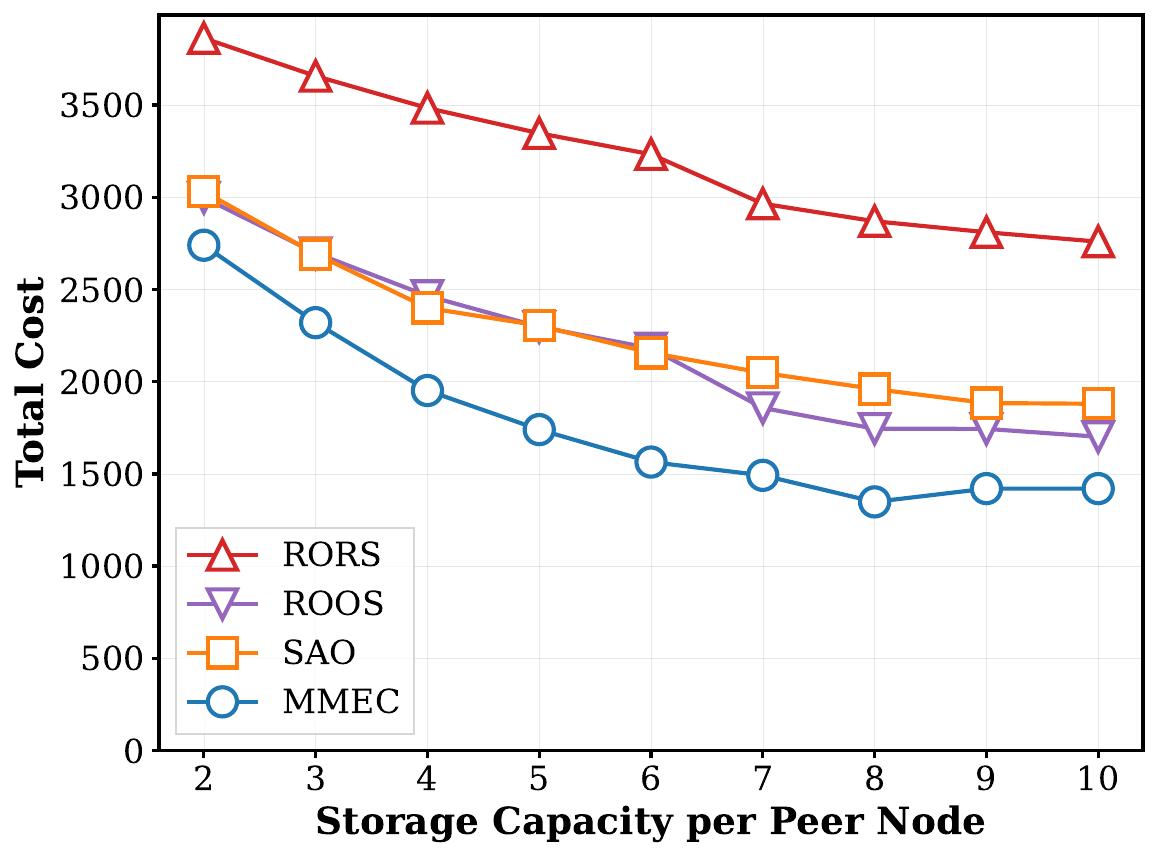}
        }   
    \caption{Performance comparison of the proposed MMEC algorithm against benchmark algorithms (RORS, ROOS, and SAO) in terms of overall transmission cost under varying system parameters.}
\end{figure*}

\subsection{Numerical Results}

\subsubsection{Impact of the Number of Users}
Fig.~\ref{fig:performance_comparison_num_users} illustrates the overall transmission cost as the number of users $U$ scales from 50 to 140. As expected, the cost for all evaluated algorithms increases monotonically with growing video demand, yet the performance gap widens significantly as the system load intensifies. The naive RORS baseline exhibits the most severe cost escalation due to random playback sequences causing severe concurrent congestion at peer nodes. Notably, the chart reveals a near-perfect performance overlap between the ROOS ablation baseline and the SAO meta-heuristic. This result indicates that the highly coupled nature of the joint video ordering and transmission scheduling problem creates a rugged solution space, trapping traditional heuristic searches (SAO) in local optima. In contrast, the proposed MMEC algorithm effectively breaks this algorithmic bottleneck, demonstrating that intelligent video ordering is just as critical as transmission scheduling. By mathematically decoupling these two phases, MMEC intelligently smooths out concurrent demand peaks to maximize the utilization of peer nodes. Consequently, MMEC achieves peak cost reductions of 67\% against RORS at $U=50$ and 36\% against SAO/ROOS at $U = 70$. While severe peer node depletion at $U=140$ inevitably forces reliance on the expensive central CDN and narrows these relative margins, MMEC still sustains significant reductions of 35\% and 17\%, respectively.

\subsubsection{Impact of the Number of Videos}
Fig.~\ref{fig:performance_comparison_num_files} illustrates the overall transmission cost as the video library size $V$ scales from 100 to 500. Across this entire range, MMEC consistently outperforms all baseline algorithms. For all evaluated approaches, the cost ultimately increases since a larger catalog under fixed peer storage capacities inevitably reduces the cache hit rate, forcing reliance on the expensive CDN node. However, MMEC reveals a distinct "cost valley" between $V=100$ and $V=300$, maintaining an absolute cost of approximately 750. This phenomenon occurs because an excessively small $V$ over-concentrates user requests, causing severe concurrent capacity bottlenecks at specific peer nodes. MMEC achieves peak cost reductions of approximately 55\% against RORS and 31\% against the SAO/ROOS cluster precisely at $V=300$. 

\subsubsection{Impact of the Number of Peer Nodes}
Fig.~\ref{fig:performance_comparison_num_servers} illustrates the overall transmission cost as the number of peer nodes $N$ scales from 10 to 100. To rigorously evaluate the algorithm's robustness against resource fragmentation, we assume a fixed total system storage and concurrent service capacity across this experiment. Consequently, increasing $N$ fragments the available resources into smaller, weaker individual nodes, which limits the serving capability per video and drives up the absolute cost across all evaluated algorithms. Notably, RORS, SAO, and ROOS exhibit a sharp cost increase as $N$ grows, indicating that they struggle significantly when the solution space becomes highly fragmented. In contrast, MMEC effectively manages this fragmentation through its mathematically guaranteed global optimization, maintaining the most stable cost trajectory. Consequently, MMEC achieves maximum relative cost reductions of approximately 59\% against RORS at $N=10$, and up to 32\% against SAO/ROOS at $N=70$.

\subsubsection{Impact of the Number of Time Slots}
Fig.~\ref{fig:performance_comparison_num_slots} illustrates the overall transmission cost as the number of time slots $T$ (representing the total video requests per user) scales from 6 to 20. The cost for all evaluated algorithms increases linearly, reflecting the fundamental scaling of system-wide demand as a larger $T$ directly elevates the total transmission volume. The naive RORS baseline exhibits the most rapid cost escalation. The SAO and ROOS baselines once again exhibit near-identical sub-optimal performance. In contrast, MMEC mathematically guarantees a globally optimal sequence regardless of the expanded temporal dimensions. Consequently, at $T=14$, MMEC achieves peak relative cost reductions of approximately 54\% against RORS, and 31\% against both the SAO and ROOS baselines.

\subsubsection{Impact of Concurrent Service Capacity per Peer Node}
Fig.~\ref{fig:performance_comparison_mcc} illustrates the overall transmission cost as the concurrent service capacity per peer node scales from 1 to 5. As expected, the cost for all evaluated algorithms decreases as capacity expands, since a higher concurrent service capacity enables more simultaneous requests to be absorbed by the peer nodes. The naive RORS baseline remains the most expensive and exhibits a relatively shallow decline. The SAO and ROOS baselines once again exhibit identical sub-optimal trajectories. In contrast, MMEC intelligently orchestrates the playback sequences to fully exploit the available capacity, maintaining the lowest absolute cost and forming a distinct plateau when the capacity exceeds 3.
Consequently, MMEC achieves maximum relative cost reductions of approximately 63\% against RORS at $d=5$, and up to 26\% against SAO/ROOS at $d=2$.

\subsubsection{Impact of Storage Capacity per Peer Node}
Fig.~\ref{fig:performance_comparison_storage} illustrates the overall transmission cost as the storage capacity per peer node scales from 2 to 10. As expected, the cost for all evaluated algorithms exhibits a downward trend, since expanded storage allows peer nodes to cache more videos, directly reducing reliance on the expensive CDN node. The naive RORS baseline consistently incurs the highest cost. Interestingly, while the SAO and ROOS baselines overlap initially, ROOS slightly outperforms SAO when storage capacity exceeds 6, indicating that the heuristic SAO struggles to navigate the enlarged solution space. In contrast, MMEC efficiently exploits the expanded storage to drive a rapid cost reduction, eventually forming a distinct plateau when capacity exceeds 8. This stabilization signals that the primary performance bottleneck has now shifted from the storage capacity of peer nodes to their concurrent service capacity. Consequently, MMEC achieves its maximum relative cost reductions peaking at approximately 53\% against RORS and 31\% against the SAO baseline at a capacity of 8, and approximately 28\% against the ROOS baseline at a capacity of 6.

\begin{figure}
\centering
\includegraphics[width=0.9\linewidth]{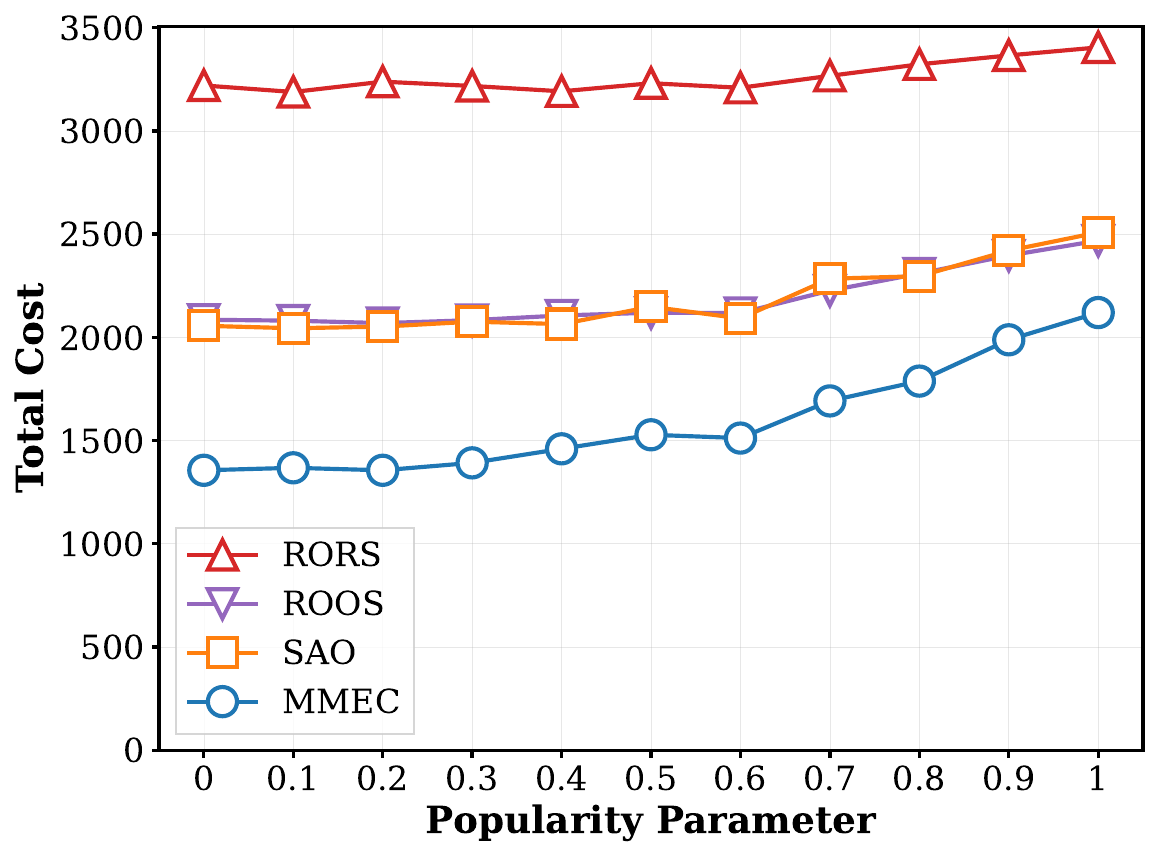}
\caption{Impact of the Zipf popularity skewness parameter $\alpha$ on the total video delivery cost across different algorithms.}
\label{fig:performance_comparison_popularity}
\end{figure}

\subsubsection{Impact of Popularity Skewness}

Fig.~\ref{fig:performance_comparison_popularity} illustrates the overall transmission cost as the Zipf popularity skewness parameter $\alpha$ scales from 0.0 to 1.0. As $\alpha$ increases, user demands become highly concentrated on a subset of hyper-popular videos. This creates severe contention for the concurrent service capacity of the specific peer nodes storing those videos, inevitably forcing the excess requests to the expensive CDN node. The naive RORS baseline consistently incurs the highest cost, while MMEC demonstrates its strongest dominance. Specifically, MMEC achieves its peak relative cost reductions at $\alpha=0.2$, reaching approximately 59\% against RORS and 36\% against both the SAO and ROOS baselines.

\section{Conclusion}
\label{conclusion}
In this paper, we addressed the challenge of reducing transmission costs in short video platforms by proposing a joint optimization framework for video ordering and transmission scheduling. Unlike traditional reactive caching mechanisms that treat user requests as rigid constraints, our approach exploits the inherent flexibility of server-driven recommendation playlist to proactively shape network demand.

By abstracting the core components of the short video delivery architecture, we modeled the short video streaming system and formulated the OVOTS problem to minimize the overall transmission cost. We provided a rigorous theoretical foundation by proving its exact reducibility to an auxiliary MCMF problem, and established a mathematically optimal mapping using K\H{o}nig's Edge Coloring Theorem. Based on these insights, we developed the MMEC algorithm, a computationally efficient, polynomial-time solution that decouples transmission scheduling from video ordering.

Extensive experimental evaluations validated that our framework effectively neutralizes flash-crowd demand spikes, maximizing the utilization of severely constrained peer nodes. The MMEC algorithm demonstrated superior performance and robustness across diverse system scales, concurrent service capacities, and storage limits. Specifically, by achieving peak cost reductions exceeding 67\% against standard baselines, our results underscore that playback sequence permutation is a highly potent, yet previously underexploited, paradigm for reducing transmission cost. Finally, potential directions for future research include extending this joint optimization framework to highly dynamic environments characterized by unpredictable user retention rates and multipath transmission.


%

\appendix

\section{Heterogeneous Peer Nodes}
\label{heterogeneous_peer_node}
In our primary performance evaluation (Section 5), we assumed a homogeneous PCDN environment to ensure mathematical clarity. However, in realistic deployments, peer nodes leased from third-party providers naturally exhibit significant heterogeneity. In this appendix, we extend our evaluation to investigate the system's performance and robustness when peer nodes possess heterogeneous resource capabilities. 

\subsection{Impact of Heterogeneous Storage Capacities}

To evaluate the robustness of our joint optimization framework in realistic, heterogeneous peer nodes, we compare the algorithmic performance under two distinct storage distributions, identical storage capacities (\textit{Uniform}) and randomly varying storage capacities (\textit{Random}), while strictly maintaining an equal total network storage capacity. We conduct 20 independent simulation trials to mitigate random variations. As illustrated in Fig.~\ref{fig:storage_allocation_comparison}, the results indicate that although storage capacities in the \textit{Random} naturally increases the overall transmission cost across all algorithms, the proposed MMEC consistently maintains its high performance. By effectively absorbing the negative impacts of heterogeneous storage distributions, MMEC demonstrates strong adaptability and robustness against the hardware disparities inherent in practical edge deployments.

\begin{figure}
    \centering
    \includegraphics[width=0.9\linewidth]{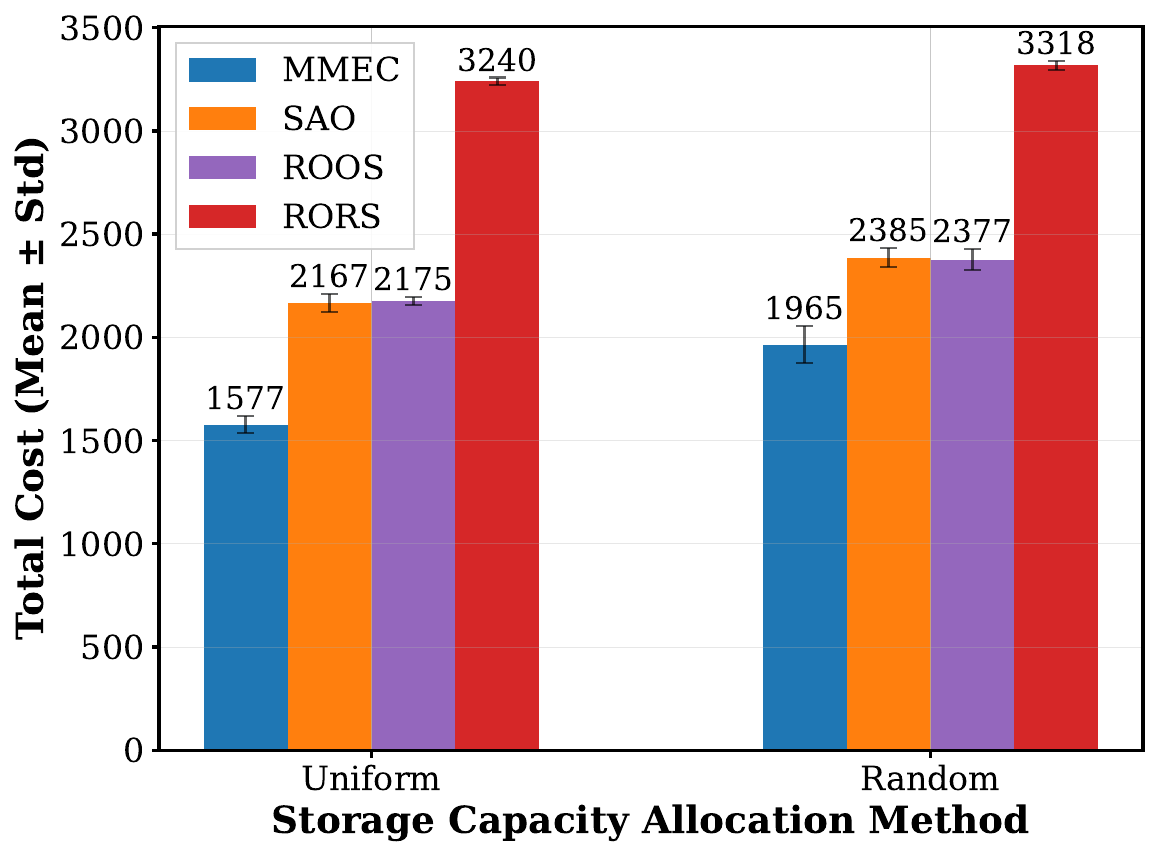}
    \caption{Performance comparison under uniform and random storage capacity allocations. }
    \label{fig:storage_allocation_comparison}
\end{figure}

\subsection{Impact of Heterogeneous Concurrent Service Capacities}

Similarly, to assess the impact of concurrent service capacities heterogeneity, we evaluate the system's performance under uniform and random concurrent service capacity allocations, strictly maintaining an identical total capacity across the peer nodes. As depicted in Fig.~\ref{fig:connection_allocation_comparison}, the results reveal that capacity heterogeneity in the \textit{Random} setting significantly elevates the overall transmission cost for all algorithms. This degradation occurs because nodes caching highly popular content might randomly be constrained by severely low concurrent limits, instantly triggering bottlenecks that force traffic to the expensive CDN. Nevertheless, the proposed MMEC algorithm consistently outperforms the baselines by intelligently permuting playback sequences to stagger user requests across the time domain. This proactive temporal staggering effectively circumvents unpredictable capacity bottlenecks, demonstrating MMEC's exceptional robustness in highly heterogeneous bandwidth environments.

\begin{figure}
    \centering
    \includegraphics[width=0.9\linewidth]{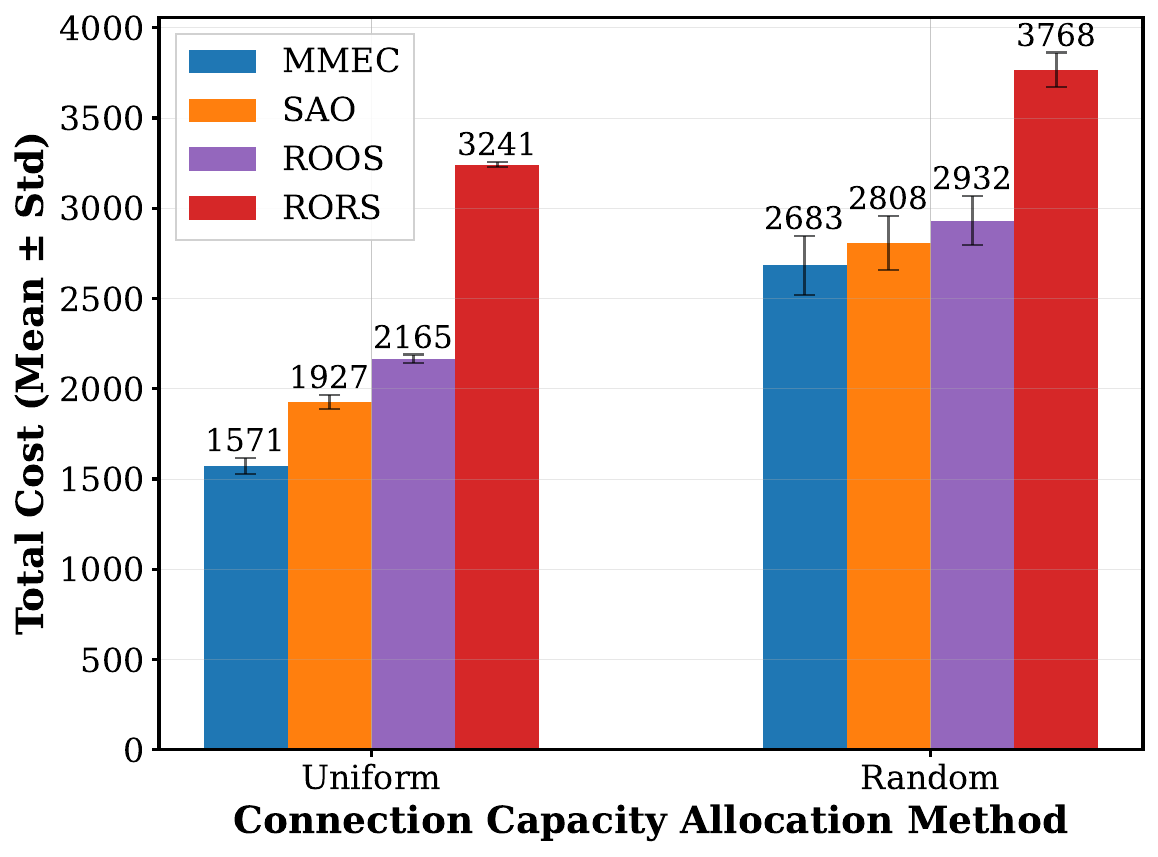}
    \caption{Performance comparison under uniform and random concurrent service capacity allocations.}
    \label{fig:connection_allocation_comparison}
\end{figure}

\section{Peer Node Storage Initialization Strategy}
\label{storage_strategy}

Finally, we investigate the impact of different peer node storage initialization strategies by comparing our \textit{Default} popularity-sorted cyclic placement against two alternative approaches: \textit{Popularity} (where each peer node independently samples videos based on the Zipf distribution) and \textit{Random} (where each node samples videos uniformly at random). As illustrated in Fig.~\ref{fig:storage_method_comparison}, the results demonstrate that both alternative strategies significantly elevate the overall transmission cost. This cost penalty arises because the \textit{Popularity} strategy induces severe redundant caching of a few hyper-popular videos, whereas the \textit{Random} strategy entirely ignores the skewed user demand, resulting in massive cache misses and forced CDN fallbacks. Nevertheless, even under these highly sub-optimal caching settings, the proposed MMEC algorithm consistently maintains the lowest overall transmission cost compared to the RORS, ROOS, and SAO baselines. 

\begin{figure}
    \centering
    \includegraphics[width=0.9\linewidth]{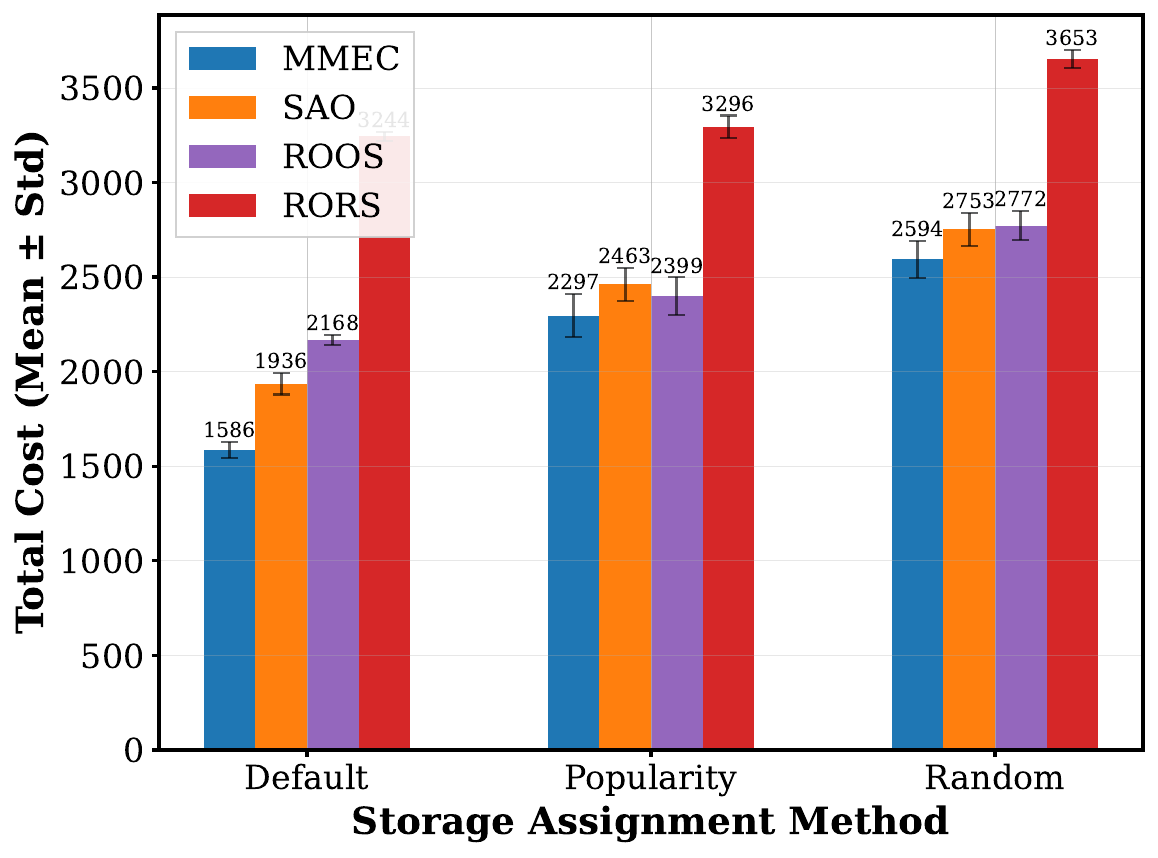}
    \caption{Performance comparison under different storage initialization strategy. }
    \label{fig:storage_method_comparison}
\end{figure}


\section*{Acknowledgment}

This work was supported by the Key Program of the National Natural Science Foundation of China (Grant No. 62531003).

\printcredits

\bibliographystyle{cas-model2-names}

\bibliography{cas-refs}

@book{ford1962flows,
  title={Flows in Networks},
  author={Ford, Lester R. and Fulkerson, Delbert R.},
  year={1962},
  publisher={Princeton University Press},
  address={Princeton, NJ}
}

@article{viola2020predictive,
  title={Predictive CDN selection for video delivery based on LSTM network performance forecasts and cost-effective trade-offs},
  author={Viola, Roberto and Martin, Angel and Morgade, Javier and Masneri, Stefano and Zorrilla, Mikel and Angueira, Pablo and Montalb{\'a}n, Jon},
  journal={IEEE Transactions on Broadcasting},
  volume={67},
  number={1},
  pages={145--158},
  year={2020}
}

@article{peroni2025end,
  title={An end-to-end pipeline perspective on video streaming in best-effort networks: a survey and tutorial},
  author={Peroni, Leonardo and Gorinsky, Sergey},
  journal={ACM Computing Surveys},
  volume={57},
  number={12},
  pages={1--47},
  year={2025}
}

@CONFERENCE{tan2019online,
  title={Online popularity prediction of video segments: Towards more efficient content delivery networks},
  author={Tan, Zhiyi and Hu, Wen and Zhang, Ya and Ding, Hao},
  booktitle={Proc. of IEEE GLOBECOM Conf.},
  pages={1--6},
  year={2019}
}

@CONFERENCE{shang2025large,
  title={A Large-scale Dataset with Behavior, Attributes, and Content of Mobile Short-video Platform},
  author={Shang, Yu and Gao, Chen and Li, Nian and Li, Yong},
  booktitle={Proc. of ACM Web Conf.},  
  pages={793--796},
  year={2025}
}

@article{huang2023digital,
  title={Digital twin based user-centric resource management for multicast short video streaming},
  author={Huang, Xinyu and Wu, Wen and Hu, Shisheng and Li, Mushu and Zhou, Conghao and Shen, Xuemin},
  journal={IEEE Journal of Selected Topics in Signal Processing},
  volume={18},
  number={1},
  pages={50--65},
  year={2023}
}

@article{nygren2010akamai,
  title={The akamai network: a platform for high-performance internet applications},
  author={Nygren, Erik and Sitaraman, Ramesh K and Sun, Jennifer},
  journal={ACM SIGOPS Operating Systems Review},
  volume={44},
  number={3},
  pages={2--19},
  year={2010}
}

@article{wei2023multipath,
  title={Multipath smart preloading algorithms in short video peer-to-peer CDN transmission architecture},
  author={Wei, Dehui and Zhang, Jiao and Li, Haozhe and Xue, Zhichen and Peng, Yajie and Han, Rui},
  journal={IEEE Network},
  volume={38},
  number={3},
  pages={285--291},
  year={2023}
}

@article{wei2025qoe,
 author={Wei, Dehui and Zhang, Jiao and Liu, Xiang and Li, Haozhe and Xue, Zhichen and Huang, Tao and Jiang, Linshan and Li, Jialin},
  journal={IEEE Journal on Selected Areas in Communications}, 
  title={QoE-Optimized MultiPath Scheduling for Video Services in Large-Scale Peer-to-Peer CDNs}, 
  year={2025},
  volume={43},
  number={2},
  pages={563-576}
}

@CONFERENCE{wei2024pscheduler,
  title={Pscheduler: Qoe-enhanced MultiPath scheduler for video services in large-scale peer-to-peer CDNs},
  author={Wei, Dehui and Zhang, Jiao and Li, Haozhe and Xue, Zhichen and Peng, Yajie and Pang, Xiaofei and Liu, Yuanjie and Han, Rui and Li, Jialin},
  booktitle={Proc. of IEEE INFOCOM Conf.},
  pages={2508--2517},
  year={2024}
}

@article{farahani2023alive,
  title={Alive: A latency-and cost-aware hybrid p2p-cdn framework for live video streaming},
  author={Farahani, Reza and {\c{C}}etinkaya, Ekrem and Timmerer, Christian and Shojafar, Mohammad and Ghanbari, Mohammad and Hellwagner, Hermann},
  journal={IEEE Transactions on Network and Service Management},
  volume={21},
  number={2},
  pages={1561--1580},
  year={2023}
}

@CONFERENCE{zhang2024enhancing,
  title={Enhancing resource management of the world's largest PCDN system for on-demand video streaming},
  author={Zhang, Rui-Xiao and Wang, Haiping and Shi, Shu and Pang, Xiaofei and Peng, Yajie and Xue, Zhichen and Liu, Jiangchuan},
  booktitle={Proc. of USENIX ATC Conf.},
  pages={951--965},
  year={2024}
}

@article{wang2024twist,
  title={Twist: A Multi-site Transmission Solution for On-demand Video Streaming},
  author={Wang, Haiping and Zhang, Ruixiao and Li, Chaojun and Xue, Zhichen and Peng, Yajie and Pang, Xiaofei and Zhang, Yixuan and Ren, Shaorui and Shi, Shu},
  journal={Proceedings of the ACM on Networking},
  volume={2},
  number={CoNEXT2},
  pages={1--19},
  year={2024}
}

@article{huang2024digital,
  title={Digital twin-based network management for better QoE in multicast short video streaming},
  author={Huang, Xinyu and Hu, Shisheng and Yang, Haojun and Wang, Xinghan and Pei, Yingying and Shen, Xuemin},
  journal={IEEE Transactions on Wireless Communications},
  year={2024},
  volume={23},
  number={11},
  pages={16187-16202}
}

@CONFERENCE{gao2022short,
  title={Short Video List Reshuffling for Minimized Wireless Resources through Video Multicast},
  author={Gao, Zhipeng and Li, Chunxi and Zhao, Yongxiang and Zhang, Baoxian and Li, Cheng},
  booktitle={Proc. of IEEE IWCMC Conf.},
  pages={943--948},
  year={2022}
}

@CONFERENCE{simonovski2025swiping,
  title={Swiping, Fast and Slow: Assessing the QoE of Short-Form Videos via Crowdsourcing},
  author={Simonovski, Filip and Hufen, Samuel and Karl, Lisa and Sayin, Alperen and Wehner, Nikolas and Ho{\ss}feld, Tobias and Seufert, Michael},
  booktitle={Proc. of IEEE QoMEX Conf.},
  pages={1--7},
  year={2025}
}

@article{gao2025startup,
  title={Startup delay aware short video ordering: Problem, model, and a reinforcement learning based algorithm},
  author={Gao, Zhipeng and Li, Chunxi and Zhao, Yongxiang and Zhang, Baoxian and Li, Cheng},
  journal={Peer-to-Peer Networking and Applications},
  volume={18},
  number={2},
  pages={74},
  year={2025}
}

@CONFERENCE{gong2022real,
  title={Real-time short video recommendation on mobile devices},
  author={Gong, Xudong and Feng, Qinlin and Zhang, Yuan and Qin, Jiangling and Ding, Weijie and Li, Biao and Jiang, Peng and Gai, Kun},
  booktitle={Proc. of ACM CIKM Conf.},
  pages={3103--3112},
  year={2022}
}

@article{quan2023alleviating,
  title={Alleviating video-length effect for micro-video recommendation},
  author={Quan, Yuhan and Ding, Jingtao and Gao, Chen and Li, Nian and Yi, Lingling and Jin, Depeng and Li, Yong},
  journal={ACM Transactions on Information Systems},
  volume={42},
  number={2},
  pages={1--24},
  year={2023}
}

@article{chen2025maximum,
  title={Maximum flow and minimum-cost flow in almost-linear time},
  author={Chen, Li and Kyng, Rasmus and Liu, Yang and Peng, Richard and Probst Gutenberg, Maximilian and Sachdeva, Sushant},
  journal={Journal of the ACM},
  volume={72},
  number={3},
  pages={1--103},
  year={2025}
}

@article{schrijver1998bipartite,
  title={Bipartite edge coloring in O($\Delta$m) time},
  author={Schrijver, Alexander},
  journal={SIAM Journal on Computing},
  volume={28},
  number={3},
  pages={841--846},
  year={1998}
}

@article{parimala2021bellman,
  title={Bellman--Ford algorithm for solving shortest path problem of a network under picture fuzzy environment},
  author={Parimala, Mani and Broumi, Said and Prakash, Karthikeyan and Topal, Selcuk},
  journal={Complex \& Intelligent Systems},
  volume={7},
  number={5},
  pages={2373--2381},
  year={2021}
}

@article{robertson1997four,
  title={The four-colour theorem},
  author={Robertson, Neil and Sanders, Daniel and Seymour, Paul and Thomas, Robin},
  journal={journal of combinatorial theory, Series B},
  volume={70},
  number={1},
  pages={2--44},
  year={1997}
}

@CONFERENCE{ghabashneh2020exploring,
  title={Exploring the interplay between CDN caching and video streaming performance},
  author={Ghabashneh, Ehab and Rao, Sanjay},
  booktitle={Proc. of IEEE INFOCOM Conf.},
  pages={516--525},
  year={2020}
}

@article{jiang2009efficient,
  title={Efficient large-scale content distribution with combination of CDN and P2P networks},
  author={Jiang, Hai and Li, Jun and Li, Zhongcheng and Bai, Xiangyu},
  journal={International Journal of Hybrid Information Technology},
  volume={2},
  number={2},
  pages={4},
  year={2009}
}

@CONFERENCE{farahani2022hybrid,
  title={Hybrid P2P-CDN architecture for live video streaming: An online learning approach},
  author={Farahani, Reza and Bentaleb, Abdelhak and {\c{C}}etinkaya, Ekrem and Timmerer, Christian and Zimmermann, Roger and Hellwagner, Hermann},
  booktitle={Proc. of IEEE GLOBECOM Conf.},
  pages={1911--1917},
  year={2022}
}

@article{zhang2014unreeling,
  title={Unreeling Xunlei Kankan: Understanding hybrid CDN-P2P video-on-demand streaming},
  author={Zhang, Ge and Liu, Wei and Hei, Xiaojun and Cheng, Wenqing},
  journal={IEEE Transactions on Multimedia},
  volume={17},
  number={2},
  pages={229--242},
  year={2014}
}

@misc{Gurobi,
  author = {Gurobi Optimization, LLC. },
  title  = {The leader in Decision Intelligence Technology},
  howpublished = "\url{https://www.gurobi.com/}",
  year = {2026},
   note = {Accessed: 2026-02-24}
}

@misc{OR-tools,
  author = {Google},
  title  = {OR-tools | google for developers},
  howpublished = "\url{https://developers.google.com/optimization}",
  year = {2026},
   note = {Accessed: 2026-02-24}
}

@misc{youtube,
  author = {YouTube},
  title  = {
YouTube},
  howpublished = "\url{https://www.youtube.com/}",
  year = {2026},
   note = {Accessed: 2026-02-24}
}

@CONFERENCE{sun2023kuaisar,
  title={KuaiSar: A unified search and recommendation dataset},
  author={Sun, Zhongxiang and Si, Zihua and Zang, Xiaoxue and Leng, Dewei and Niu, Yanan and Song, Yang and Zhang, Xiao and Xu, Jun},
  booktitle={Proc. of ACM CIKM Conf.},
  pages={5407--5411},
  year={2023}
}

@misc{netflix,
  author = {Netflix},
  title  = {Unlimited movies, TV shows, and more},
  howpublished = "\url{https://www.netflix.com/}",
  year = {2026},
   note = {Accessed: 2026-02-24}
}

@misc{kwai,
  author = {kwai},
  title  = {Kwai is a social network for short videos and trends. },
  howpublished = "\url{https://www.kwai.com/}",
  year = {2026},
   note = {Accessed: 2026-02-24}
}

@misc{pptv,
  author = {PPTV},
  title  = {PPTV },
  howpublished = "\url{https://www.pptv.com/}",
  year = {2026},
   note = {Accessed: 2026-02-24}
}

@article{bag2019efficient,
  title={An efficient recommendation generation using relevant Jaccard similarity},
  author={Bag, Sujoy and Kumar, Sri Krishna and Tiwari, Manoj Kumar},
  journal={Information Sciences},
  volume={483},
  number = {C},
  pages={53--64},
  year={2019}
}

@CONFERENCE{liu2011novasky,
  title={Novasky: Cinematic-quality VoD in a P2P storage cloud},
  author={Liu, Fangming and Shen, Shijun and Li, Bo and Li, Baochun and Yin, Hao and Li, Sanli},
  booktitle={Proc. of IEEE INFOCOM Conf.},
  pages={936--944},
  year={2011}
}

@CONFERENCE{jung2002flash,
  title={Flash crowds and denial of service attacks: Characterization and implications for CDNs and web sites},
  author={Jung, Jaeyeon and Krishnamurthy, Balachander and Rabinovich, Michael},
  booktitle={Proc. of ACM WWW Conf.},
  pages={293--304},
  year={2002}
}








\end{document}